\documentclass[a4paper,UKenglish]{lipics-v2016}
\pdfoutput=1
\usepackage{microtype}

\usepackage{enumerate}
\usepackage{amssymb}

\usepackage{xcolor,latexsym,amsmath,extarrows,alltt}
\usepackage{xspace}
\usepackage{booktabs}
\usepackage{mathtools}
\usepackage{enumitem}
\usepackage{stmaryrd}
\usepackage{microtype}
\usepackage{wrapfig}
\usepackage{mathabx}

\theoremstyle{definition}\newtheorem{algorithm}[theorem]{Algorithm}

\newcommand{\nats}[0]{\ensuremath{\mathbb{N}}}

\newcommand{\timecomp}[1]{\ensuremath{\textrm{TIME}\left(#1\right)}}
\newcommand{\etime}[1]{\ensuremath{\textrm{E}^{#1}\textrm{TIME}}}
\newcommand{\exptime}[1]{\ensuremath{\textrm{EXP}^{#1}\textrm{TIME}}}

\newcommand{\B}{\mathcal{B}}
\newcommand{\F}{\mathcal{F}}
\renewcommand{\P}{\mathbb{P}}
\newcommand{\V}{\mathcal{V}}
\newcommand{\Flab}{\mathcal{F}_{\mathtt{lab}}}
\newcommand{\NF}{\mathcal{N}\!\mathcal{F}}
\newcommand{\Sorts}{\mathcal{S}}
\newcommand{\Constructors}{\mathcal{C}}
\newcommand{\Defineds}{\mathcal{D}}
\newcommand{\Data}{\mathcal{D}\!\!\mathcal{A}}
\newcommand{\Terms}{\mathcal{T}}
\newcommand{\Rules}{\mathcal{R}}
\newcommand{\Ruleslab}{\Rules_{\mathtt{lab}}}
\newcommand{\FV}{\mathit{FV}}

\newcommand{\asort}{\iota}
\newcommand{\bsort}{\kappa}
\newcommand{\atype}{\sigma}
\newcommand{\btype}{\tau}
\newcommand{\ctype}{\pi}

\newcommand{\clausevar}{\texttt{(var)}}
\newcommand{\clauseapp}{\texttt{(app)}}
\newcommand{\clauseabs}{\texttt{(abs)}}
\newcommand{\clausefun}{\texttt{(fun)}}

\newcommand{\clausebeta}{\texttt{(beta)}}

\newcommand{\app}[2]{#1\cdot #2}
\newcommand{\apps}[3]{#1 \cdot #2 \cdots #3}
\newcommand{\abs}[2]{\lambda #1.#2}
\newcommand{\symb}[1]{\mathtt{#1}}
\newcommand{\unknown}[1]{\underline{\symb{#1}}}
\newcommand{\encode}[1]{\overline{\symb{#1}}}
\newcommand{\interpret}[1]{\llbracket #1 \rrbracket}
\newcommand{\numinterpret}[1]{\langle #1 \rangle}
\newcommand{\domain}{\mathtt{domain}}
\newcommand{\labl}{\mathsf{label}}
\newcommand{\numrep}[1]{[#1]}
\newcommand{\Conf}{\mathsf{Confirmed}}

\newcommand{\arrtype}{\Rightarrow}
\newcommand{\arrz}{\to}
\newcommand{\arrzt}{\Rightarrow}
\newcommand{\arr}[1]{\arrz_{#1}}
\newcommand{\arrr}[1]{\arr{#1}^*}
\newcommand{\arrp}[1]{\arr{#1}^+}
\newcommand{\subterm}{\lhd}
\newcommand{\supterm}{\rhd}
\newcommand{\subtermeq}{\unlhd}
\newcommand{\suptermeq}{\unrhd}
\newcommand{\transition}[5]{#1~\displaystyle{\mathop{=\!\!=\!\!\!
  \Longrightarrow}^{#2/#3\ #4}}~#5}
\newcommand{\card}{\mathtt{card}}

\newcommand{\blank}{\textbf{\textvisiblespace}}
\newcommand{\nul}{\symb{0}}
\newcommand{\one}{\symb{1}}
\newcommand{\nil}{\ensuremath{\rhd}}
\newcommand{\strue}{\symb{true}}
\newcommand{\sfalse}{\symb{false}}

\newcommand{\bits}{\symb{string}}
\newcommand{\bool}{\symb{bool}}

\newcommand{\arxivversion}[1]{#1}
\newcommand{\fscdversion}[1]{}

\newcommand{\CKchange}[1]{\textcolor{blue}{#1}}
\renewcommand{\CKchange}[1]{#1}

\title{Complexity Hierarchies and Higher-Order Cons-Free Rewriting%
\footnote{The authors are supported by the Marie Sk{\l}odowska-Curie
action ``HORIP'', program H2020-MSCA-IF-2014, 658162 and by the Danish
Council for Independent Research Sapere Aude grant ``Complexity via
Logic and Algebra'' (COLA).}}
\titlerunning{Complexity Hierarchies and Higher-Order Cons-Free Rewriting}

\author[]{Cynthia Kop}
\author[]{Jakob Grue Simonsen}
\affil{Department of Computer Science, University of Copenhagen (DIKU)\\
  \texttt{\{kop,simonsen\}@di.ku.dk}}

\authorrunning{C. Kop and J.\,G. Simonsen}

\Copyright{Cynthia Kop and Jakob Grue Simonsen}

\EventEditors{Brigitte Pientka and Delia Kesner}
\EventNoEds{2}
\EventLongTitle{1st International Conference on Formal Structures for Computation and Deduction (FSCD 2016)}
\EventShortTitle{FSCD 2016}
\EventAcronym{FSCD}
\EventYear{2016}
\EventDate{June 22--26, 2016}
\EventLocation{Porto, Portugal}
\EventLogo{}
\SeriesVolume{52}
\ArticleNo{49}
\DOI{10.4230/LIPIcs.FSCD.2016}

\begin{document}

\maketitle

\begin{abstract}
Constructor rewriting systems are said to be cons-free if, roughly,
constructor terms in the right-hand sides of rules are subterms of
constructor terms in the left-hand side; the computational intuition
is that rules cannot build new data structures. It is well-known that
cons-free programming languages can be used to characterize
computational complexity classes, and that cons-free first-order term
rewriting can be used to characterize the set of polynomial-time
decidable sets.

We investigate cons-free higher-order term rewriting systems, the
complexity classes they characterize, and how these depend on the
order of the types used in the systems. We prove that, for every $k
\geq 1$, left-linear cons-free systems with type order $k$
characterize $\etime{k}$ if arbitrary evaluation is used
(i.e., the system does not have a fixed reduction strategy).

The main difference with prior work in implicit complexity is that (i) our results hold for
non-orthogonal term rewriting systems with possible rule overlaps
with no assumptions about reduction strategy, (ii) results for such
term rewriting systems have previously only been obtained for $k=1$, and with additional syntactic restrictions on top
of cons-freeness and left-linearity.

Our results are apparently among the first implicit
characterizations of the hierarchy $\textrm{E} = \etime{1} \subsetneq \etime{2} \subsetneq \cdots$.
Our work confirms prior results that having full non-determinism (via overlaps of rules) does not directly
allow characterization of non-deterministic complexity classes like $\textrm{NE}$.  We also show
that non-determinism makes the classes characterized highly sensitive to minor syntactic changes such as admitting product types or
non-left-linear rules.
\end{abstract}

\section{Introduction}

In~\cite{jon:01}, Jones introduces \emph{cons-free
programming}: working with a small functional programming language,
cons-free programs are exactly those where function bodies cannot contain use of data constructors (the ``cons'' operator on lists).  Put differently, a cons-free
program is \emph{read-only}: data structures cannot be created or
altered, only read from the input; and any data passed as arguments to recursive
function calls must thus be part of the original input.

The interest in such programs lies in their applicability to
computational complexity: by imposing
cons-freeness, the resulting set of programs can only decide the sets in 
a proper subclass of the Turing-decidable sets, indeed are said to
\emph{characterize} the subclass.  
Jones goes on to show that adding further restrictions such as type order
or enforcing tail recursion lowers the resulting expressiveness to known
classes.  For example, cons-free programs with data order $0$ can
decide exactly the sets in PTIME, while tail-recursive cons-free
programs with data order $1$ can decide exactly the sets in PSPACE.
The study of such restrictions and the complexity classes characterized
is a research area known as \emph{implicit complexity}
and has a long history with many distinct approaches
(see, e.g., \cite{DBLP:conf/tlca/Baillot07,baillot_et_al:LIPIcs:2012:3664,DBLP:conf/esop/BaillotGM10,DBLP:journals/cc/BellantoniC92,DBLP:journals/apal/BellantoniNS00,Hofmann:hab,DBLP:journals/tcs/KristiansenN04}).

Rather than a toy language, it 
is tantalizing to consider \emph{term rewriting} instead.  Term
rewriting systems have no fixed evaluation order (so call-by-name or
call-by-value can be introduced as needed, but are not
\emph{required});
and term rewriting is natively non-deterministic,
allowing distinct rules to be applied (``functions to be invoked'')
to the same piece of syntax,
hence could be useful for extensions towards
non-deterministic complexity classes. Implicit complexity using
term rewriting has seen significant advances using a plethora
of approaches (e.g.\ \cite{DBLP:conf/aplas/AvanziniEM12,DBLP:conf/rta/AvanziniM10,DBLP:journals/corr/AvanziniM13}).
Most of this research has, however, considered fixed evaluation orders (most prominently innermost reduction),
and if not, then systems which are either orthogonal, or at least confluent (e.g. \cite{DBLP:conf/rta/AvanziniM10}).
Almost all of the work considers only first-order rewriting.

The authors of~\cite{car:sim:14} provide a first definition of
cons-free term rewriting without constraints on evaluation order or confluence requirements, and prove that this class---limited to \emph{first-order} rewriting---characterizes $\textrm{PTIME}$.  However,
 they impose a rather severe partial linearity restriction on the
programs.  
This paper seeks to answer two questions: (i) what happens if
\emph{no} restrictions beyond left-linearity and cons-freeness are
imposed? And (ii) what if \emph{higher-order} term
rewriting---including bound variables as in the lambda
calculus---is allowed?
We obtain that $k^{\text{th}}$-order
cons-free term rewriting exactly characterizes $\etime{k}$.  This is
surprising because in Jones' rewriting-like language,
$k^{\text{th}}$-order programs characterize $\exptime{k-1}$:
surrendering both determinism and evaluation order thus significantly
increases expressivity.

\fscdversion{An extended version, including appendices with full proofs,
is available online~\cite{fullversion}.}
\arxivversion{Note that an appendix containing full proofs is included at the end
of the paper.}

\section{Preliminaries}\label{sec:prelims}

\subsection{Computational complexity}

We presuppose introductory working knowledge of computability and complexity theory (corresponding to standard textbooks, e.g., \cite{Jones:CompComp}). Notation is fixed below.

Turing Machines (TMs)
are triples $(A,S,T)$ where
$A$ is a finite set of tape symbols such that $A \supseteq I \cup
\{\blank\}$, where $I \supseteq \{ 0,1 \}$ is a set of \emph{initial
symbols} and $\blank \notin I$ is the special \emph{blank} symbol;
$S \supseteq \{ \symb{start}, \symb{accept}, \symb{reject} \}$ is a
 finite set of states, and 
 $T$ is a finite set of transitions $(i,r,w,d,j)$
  with
$i \in S \setminus \{\symb{accept},\symb{reject}\}$ (the \emph{original state}),
$r \in A$ (the \emph{read symbol}), $w \in A$ (the \emph{written symbol}),
 $d \in \{ \symb{L},\symb{R} \}$ (the \emph{direction}), and
$j \in S$ (the \emph{result state}).
We sometimes write this transition as $\transition{i}{r}{w}{d}{j}$.
All TMs in the paper are deterministic and (which
we can assume wlog.) do not get stuck: for every pair $(i,r)$ with
$i \in S\setminus \{\symb{accept},\symb{reject}\}$ and $r \in A$ there
is exactly one transition $(i,r,w,d,j)$.
Every TM has a single, right-infinite tape.  
  
A \emph{valid tape} is a right-infinite sequence of tape symbols with only
finitely many not $\blank$.  A \emph{configuration} of a TM
is a triple $(t,p,s)$ with $t$ a valid tape, $p \in \nats$ and $s \in S$. The transitions $T$
induce a binary relation $\arrzt$ between configurations in the
obvious way.

A TM with input alphabet $I$ \emph{decides}
$X \subseteq I^+$ if for any string $x \in I^+$, we have 
$x \in X$ if{f}
$(\blank x_1\dots x_n\blank\blank\dots,0,\symb{start}) \arrzt^*
(t,i,\symb{accept})$ for some $t,i$, and $(\blank x_1\dots x_n\blank
\blank\dots,0,\symb{start})\linebreak %OVERFULL HBOX AVOIDANCE
\arrzt^* (t,i,\symb{reject})$ otherwise
(i.e., the machine halts on all inputs, ending in %the
$\symb{accept}$ or $\symb{reject}$ %state
depending on whether $x \in X$).
If $f : \nats \longrightarrow \nats$ is a function, a (deterministic) TM
\emph{runs in time} $\lambda n.f(n)$ if, for each $n \in \nats \setminus \{0\}$ and each $x \in I^n$:
%, we have
$(\blank x\blank\blank\dots,0,\symb{start}) \arrzt^{\leq f(n)} (t,i,\unknown{s})$ for $\unknown{s} \in \{\symb{accept},\symb{reject}\}$, where $\arrzt^{\leq f(n)}$ denotes a sequence of at most
$f(n)$ transitions.

\paragraph*{Complexity and the ETIME hierarchy}

For $k,n \geq 0$, let $\mathrm{exp}_2^0(n) = n$ and
$\mathrm{exp}_2^{k+1}(n) = 2^{\mathrm{exp}_2^k(n)} =
\mathrm{exp}_2^k(2^n)$.

\begin{definition}
Let $f : \nats \longrightarrow \nats$ be a function. Then, $\timecomp{f(n)}$ is the set of all
 $S \subseteq \{0,1\}^+$ such that there exist
 $a > 0$ and a deterministic TM running in time $\lambda n.a \cdot f(n)$ that decides $S$ (i.e., $S$ is decidable in time $O(f(n))$).
For $k \geq 1$ define:
$
\etime{k} \triangleq \bigcup_{a \in \nats} \timecomp{\textrm{exp}_2^{k}(an)}
$
\end{definition}

Observe in particular that $\etime{1} = \bigcup_{a \in \nats} \timecomp{\textrm{exp}_2^{1}(an)} = \bigcup_{a \in \nats} \timecomp{2^{an}} = \textrm{E}$ (where $\textrm{E}$ is the usual complexity class of this name, see e.g., \cite[Ch.\ 20]{Papadimitriou:complexity}).

Note that for any $d,k \geq 1$, we have
$(\mathrm{exp}_ 2^k(x))^d = 2^{d \cdot \mathrm{exp}_2^{k-1}(x)
} \leq 2^{\mathrm{exp}_2^{k-1}(dx)} = \mathrm{exp}_2^k(dx)$. Hence, if $P$ is a polynomial with non-negative
integer coefficients and the set $S \subseteq \{0,1\}^+$ is decided by an algorithm running in time $O(P(\mathrm{exp}_ 2^k(an)))$ for some $a \in \nats$, then $S \in \etime{k}$.

Using the Time Hierarchy Theorem \cite{Sipser:comp}, it is easy to see that
$\textrm{E} = \etime{1} \subsetneq \etime{2} \subsetneq \etime{3} \subsetneq \cdots$.
The union
$\bigcup_{k \in \nats} \etime{k}$ is the set $\textrm{ELEMENTARY}$ of elementary languages.

\subsection{Higher-order rewriting}

Unlike first-order term rewriting, there is no single, unified approach to
higher-order term rewriting, but rather a number of different co-extensive systems
with distinct syntax; for an overview of basic issues, see \cite{T03_R}.
\CKchange{We} will use \emph{Algebraic Functional Systems} (AFSs)
\cite{jou:rub:99,bla:jou:rub:08}\CKchange{, in the simplest form
(which disallows partial applications)}.
However, \CKchange{our proofs do not use any} particular features of
AFSs that preclude using different higher-order formalisms.
%We use the simplest variation of AFS% (which we call \emph{base AFSs})
%; these admit neither product types nor partial applications.  

\paragraph*{Types and Terms}

We assume a non-empty set $\Sorts$ of \emph{sorts}, and define
types and type orders as follows:
(i) every $\asort \in \Sorts$ is a type of order \emph{0}; (ii) if $\atype,\btype$ are types of order $n$ and $m$ respectively, then $\atype \arrtype \btype$ is a type of order $\max(n+1,m)$.
Here
$\arrtype$ is right-associative, so $\atype \arrtype \btype \arrtype
\ctype$ should be read $\atype \arrtype (\btype \arrtype \ctype)$.
A \emph{type declaration} of order $k \geq 0$ is a tuple $[\atype_1 \times
\dots \times \atype_n] \arrtype \asort$ with all $\atype_i$ types of
order at most $k-1$, and $\asort \in \Sorts$; if $n = 0$ this
declaration may simply be denoted $\asort$.

We additionally assume given disjoint sets $\F$ of \emph{function
symbols} and $\V$ of \emph{variables}.  Each symbol in $\F$
is associated with a unique type declaration, and each variable in
$\V$ with a unique type.
The set $\Terms(\F,\V)$ of \emph{terms over $\F$ and $\V$} consists of
those expressions $s$ such that $\vdash s : \atype$ can be
derived for some type $\atype$ using the following clauses:

\begin{tabular}{lllll}
  \clausevar &
  & $\vdash x : \atype$
  & \mbox{} & if $x:\atype \in \V$ \\
  \clauseapp &
  & $\vdash s \cdot t : \btype$
  & \mbox{}& if $s : \atype \arrtype \btype$ and
      $t : \atype$  \\
  \clauseabs &
  & $\vdash \abs{x}{s} : \atype \arrtype \btype$
  & \mbox{} & if $x:\atype \in \V$ and $s : \btype$ \\
  \clausefun &
  & $\vdash f (s_1 , \ldots , s_n) : \asort$
  & \mbox{} &
  if $f : {[\atype_1 \times \ldots \times \atype_n] \arrtype \asort} \in \F$
  and
  $s_1 :\atype_1 , \ldots, s_n:\atype_n$
\end{tabular}

Clearly, each term has a \emph{unique} type.
Note that a function symbol $f : [\atype_1 \times \ldots \times
\atype_n] \arrtype \asort$ takes exactly $n$ arguments, and its output
type $\asort$ is a sort.  The \emph{abstraction} construction
$\abs{x}{s}$ binds occurrences of $x$ in $s$ as in
the $\lambda$-calculus, and $\alpha$-conversion is defined for terms \emph{mutatis mutandis};  we identify
terms modulo $\alpha$-conversion, renaming bound variables if necessary.
Application is left-associative.
The set of variables of $s$ which are not bound is denoted
$\FV(s)$.  A term $s$ is \emph{closed} if $\FV(s) =
\emptyset$.
We say that a term $s$ \emph{has base type}
if $\vdash s : \asort \in \Sorts$.

\begin{example}\label{ex:bitlists}
We will often use extensions of the signature
$\F_{\bits}$, given by:
\[
\begin{array}{rclrclrcl}
\strue & : & \bool &
\phantom{XY}\nul & : & [\bits] \arrtype \bits &
\phantom{XY}\nil & : & \bits \\
\sfalse & : & \bool &
\one & : & [\bits] \arrtype \bits \\
\end{array}
\]
Terms are for instance $\strue,\ \abs{x}{\nul(\one(x))}$ and
$\app{(\abs{x}{\nul(x)})}{\one(y)}$.  The first and last
of these terms have
base type, and the first two are closed; the last one has $y$ as a
free variable.
\end{example}
A \emph{substitution} is a type-preserving map from $\V$ to
$\Terms(\F,\V)$ which is the identity on all but finitely many
variables.  Substitutions
$\gamma$ are extended to arbitrary terms $s$, notation $s\gamma$,
by using $\alpha$-conversion to rename all bound variables in $s$ to
fresh ones, then replacing each unbound variable $x$ by $\gamma(x)$.
A \emph{context} $C[]$ is a
term in $\Terms(\F,\V)$ in which a single occurrence of a variable
%has been
\CKchange{is}
replaced by
%a formal symbol
\CKchange{a symbol}
$\Box \notin \F \cup \V$.
The result of replacing $\Box$ in $C[]$ by a term $s$
(of matching type) is denoted $C[s]$.
Free variables may be captured; e.g.
$(\abs{x}{\Box})[x] =\abs{x}{x}$.
If $s = C[t]$ we say that $t$ is a \emph{subterm} of
$s$, notation $t \subtermeq s$, or $t \subterm s$ if $C[] \neq \Box$.

\paragraph*{Rules and Rewriting}
A \emph{rule} is a pair $\ell \arrz r$ of terms in $\Terms(\F,\V)$
with the same \emph{sort} (i.e.~%there is $\asort \in \Sorts$ with
$\vdash \ell : \asort$ and $\vdash r : \asort$
\CKchange{for some $\asort \in \Sorts$}), such that $\ell$ has
the form $f(\ell_1,\dots,\ell_n)$ with $f \in \F$
and such that $\FV(r) \subseteq \FV(\ell)$.
A rule $\ell \arrz r$ is \emph{left-linear} if every variable occurs at most once in $\ell$.
We assume given a set $\Rules$ of rules, and define the one-step rewrite
relation $\arr{\Rules}$ on $\Terms(\F,\V)$ as follows:

\begin{tabular}{rcll}
$C[\ell\gamma]$ & $\arr{\Rules}$ & $C[r\gamma]$ &
  with $\ell \arrz r \in \Rules$, $C$ a context, $\gamma$ a
  substitution \\
$C[(\abs{x}{s}) \cdot t]$ &
  $\arr{\Rules}$ & $C[s[x:=t]]$ & \\
\end{tabular} \\
We may write $s \arr{\beta} t$ for a rewrite step using \clausebeta.
Let $\arrp{\Rules}$ denote the transitive closure of $\arr{\Rules}$
and $\arrr{\Rules}$ the transitive-reflexive closure.
We say that \emph{$s$ reduces to $t$} if $s \arr{\Rules} t$.
A term $s$ is in
\emph{normal form} if there is no $t$ such that $s \arr{
\Rules} t$,
and \emph{$t$ is a normal form of $s$} if $s \arrr{\Rules}
t$ and $t$ is in normal form.
An AFS is a pair $(\F,\Rules)$, generating a set of terms and a
reduction relation.
The \emph{order} of an AFS is the maximal
order of any type declaration in $\F$.

\begin{example}\label{ex:plusfo}
Recall the signature $\F_\bits$ from Example~\ref{ex:bitlists};
let $\F_{\symb{count}}$ be its extension with $\symb{succ} :
[\bits] \arrtype \bits$.  We consider the AFS $(\F_{\symb{count}},
\Rules_{\symb{count}})$ with the following rules:
\[
\begin{array}{crclcrcl}
\mathsf{(A)} & \symb{succ}(\nil) & \arrz & \one(\nil) &
\mathsf{(B)} & \symb{succ}(\nul(xs)) & \arrz & \one(xs) \\
& & & &
\mathsf{(C)} & \symb{succ}(\one(xs)) & \arrz & \nul(\symb{succ}(xs)) \\
\end{array}
\]
This is a \emph{first-order} AFS, implementing the successor function
on a binary number expressed as a bitstring with the least significant
digit first.  For example, %the number
%$11$ is represented by
%%the string
%$\one(\one(\nul(\one(\nil))))$, and indeed
%$\symb{succ}(\one(\one(\nul(\one(\nil))))) \arrr{\Rules}
%%\nul(\symb{succ}(\one(\nul(\one(\nil))))) \arr{\Rules}
%%\nul(\nul(\symb{succ}(\nul(\one(\nil))))) \arr{\Rules}
%\nul(\nul(\one(\one(\nil))))$ which represents $12$.
\CKchange{$5$ is represented by $\one(\nul(\one(\nil)))$, and indeed
$\symb{succ}(\one(\nul(\one(\nil)))) \arr{\Rules}
\nul(\symb{succ}(\nul(\one(\nil)))) \arr{\Rules}
\nul(\one(\one(\nil)))$, which represents $6$.}
\end{example}

\begin{example}\label{ex:plusho}
Alternatively, we may define a bit-sequence as a \emph{function}: let
$\F_{\symb{hocount}}$ be the extension of $\F_\bits$ with
$\symb{not} : [\bool] \arrtype \bool,\ 
\symb{ite} : [\bool
\times \bool \times \bool]  \arrtype \bool$ and $\symb{all},\ 
\symb{succ} : [(\bool \arrtype \bool) \times \bits] \arrtype \bits$.
Let $\Rules_{\symb{hocount}}$ consist of:
\[
\begin{array}{crclcrcl}
\mathsf{(A)} & \symb{ite}(\strue,x,y) & \arrz & x &
\mathsf{(C)} & \symb{not}(x) & \arrz & \symb{ite}(x,\sfalse,\strue) \\
\mathsf{(B)} & \symb{ite}(\sfalse,x,y) & \arrz & y &
\mathsf{(D)} & \symb{all}(F,\nil) & \arrz & \app{F}{\nil} \\
\mathsf{(E)} & \symb{all}(F,\unknown{a}(xs)) & \arrz &
  \multicolumn{5}{l}{
    \symb{ite}(\app{F}{\unknown{a}(xs)},\symb{all}(F,xs),\sfalse)
    \hfill\llbracket\text{for}\ \unknown{a} \in \{\nul,\one\}\rrbracket
  } \\
\mathsf{(F)} & \symb{succ}(F,\nil) & \arrz &
  \multicolumn{5}{l}{\symb{not}(\app{F}{\nil})} \\
\mathsf{(G)} & \symb{succ}(F,\unknown{a}(xs)) & \arrz &
  \multicolumn{5}{l}{
    \symb{ite}(\symb{all}(F,xs),\symb{not}(\app{F}{\unknown{a}(xs)}),
    \app{F}{\unknown{a}(xs)})\ \llbracket\text{for}\ \unknown{a} \in
    \{\nul,\one\}\rrbracket
  } \\
\end{array}
\]
Note that $\mathsf{(E)}$ and $\mathsf{(G)}$ each represent \emph{two}
rules: one for each choice of $\unknown{a}$.
This AFS is second-order, due to
%the $\symb{all}$ and $\symb{succ}$
%symbols with argument type $\bool \arrtype \bool$.
\CKchange{$\symb{all}$ and $\symb{succ}$}.
A function $F$
represents a (potentially infinite) binary number, with the
$i^{\text{th}}$ bit given by $\app{F}{t}$ \CKchange{for any bitstring $t$}
of length $i$ (counting from $i = 0$, \CKchange{so $t = \nil$}).  Thus,
the number $0$ is represented by\CKchange{, e.g.,} %for instance the term
$\abs{x}{\sfalse}$, and $1$ by $\symb{ONE} ::=
\abs{x}{\symb{succ}(\abs{y}{\sfalse},x)}$.  %Then
\CKchange{Indeed}
$\app{\symb{ONE}}{\nil} = \app{(\abs{x}{\symb{succ}(\abs{y}{\sfalse},
x)})}{\nil} \arr{\beta} \symb{succ}(\abs{y}{\sfalse},\nil)
\arr{\Rules} \symb{not}(\app{(\abs{y}{\sfalse})}{\nil}) \arr{\beta}
\symb{not}(\sfalse) \arr{\Rules} \strue$, and
$\app{\symb{ONE}}{\nul^k(\nil)} \arrr{\Rules} \sfalse$ for $k > 0$.
\end{example}

We fix a partitioning of $\F$ into two disjoint sets, $\Defineds$ of
\emph{defined symbols} and $\Constructors$ of \emph{constructor
symbols}, such that $f \in \Defineds$ for all $f(\vec{\ell}) \arrz r
\in \Rules$.
A term $s$ is a \emph{constructor term} if it is in
$\Terms(\Constructors,\V)$ and a \emph{proper constructor term} if it
also contains no applications or abstractions.
A \emph{closed} proper constructor term is also called a \emph{data
term}.
The set of data terms is denoted $\Data$.
Note that data terms are built using only clause \clausefun.
A term $f(s_1,\dots,s_n)$ with $f \in \Defineds$ and each $s_i \in
\Data$ is called a \emph{basic term}.
A \emph{constructor rewriting system} is an AFS where each rule
$f(\ell_1,\dots,\ell_n) \arrz r \in \Rules$ satisfies that all
$\ell_i$ are proper constructor terms (and $f \in \Defineds$).
An AFS is a \emph{left-linear constructor rewriting system} if moreover each rule is left-linear.

In a constructor rewriting system, $\beta$-reduction steps
can always be done prior \CKchange{to %performing
other} steps: if $s$ has a normal form $q$ and $s
\arr{\beta} t$, then also $t \arrr{\Rules} q$.  Therefore we can
\CKchange{(and will!)}
safely assume %(and will do so)
that the right-hand sides of %all
rules
%are $\beta$-normal (i.e.~cannot be reduced with $\arr{\beta}$).
\CKchange{are in normal form with respect to $\arr{\beta}$}.

\begin{example}
The AFSs from Examples~\ref{ex:plusfo} and~\ref{ex:plusho} are
left-linear constructor rewriting systems.  In
Example~\ref{ex:plusfo}, $\Constructors = \F_\bits$ and $\Defineds =
\{\symb{succ}\}$.  If a rule $\nul(\nil) \arrz \nil$ were added to
$\Rules_{\symb{count}}$, it would no longer be a constructor system,
as this would force $\nul$ to be in $\Defineds$, conflicting
with rule $\mathsf{(B)}$.
A rule such as $\symb{equal}(xs,xs) \arrz \strue$ would break
left-linearity.
\end{example}

\begin{remark}
Constructor rewriting systems---typically
left-linear---are very common both in the literature on term
rewriting and in functional programming, where similar restrictions
are imposed.  Left-linear systems are
well-behaved: contraction of non-overlapping redexes cannot 
destroy redexes that they themselves are arguments of.  Constructor
systems avoid non-root overlaps, and allow for a clear split between
data and intermediate terms.

They are, however, less common in the literature on \emph{higher-order} term
rewriting, and the notion of a \emph{proper} constructor term is new
for AFSs (although the exclusion of abstractions and applications in
the left-hand sides roughly corresponds to \emph{fully extended
pattern HRSs} in Nipkow's style of higher-order
rewriting~\cite{may:nip:98}).  
%Intuitively, this
%definition extends the first-order notion of a constructor term by
%observing that applications act as a defined symbol for
%$\beta$-reduction, and abstractions can be seen as anonymous
%functions.
\end{remark}

%\vspace{-0.4cm}

\paragraph*{Deciding problems using rewriting}

Like Turing Machines, an AFS can decide a set $X \subseteq I^+$
(where $I$ is a finite set of symbols).
Consider AFSs with a signature $\F = \Constructors
\cup \Defineds$ where $\Constructors$ contains symbols $\nil :
\bits,\ \strue : \bool,\ \sfalse : \bool$ and $\symb{a} : [\bits]
\arrtype \bits$ for all $a \in I$.
There is an \CKchange{obvious} %computable
correspon\-dence between elements of
$I^+$ and data terms of sort $\bits$; if $x \in I^+$, we write
$\encode{x}$ for the corresponding data term.
The AFS \emph{accepts} $D \subseteq I^+$ if
there is a designated \CKchange{defined} symbol
$\symb{decide} : [\bits] \arrtype \bool$ such
that, for every $x \in I^+$ we have $\symb{decide}(\encode{x})
\arrr{\Rules} \symb{true}$ if{f} $x \in D$.
More generally, we are interested in the reductions of a given
\emph{basic term} to a \emph{data term}.

We use the acceptance criterion above---reminiscent of the acceptance criterion of non-deterministic Turing machines---because term rewriting is inherently non-deterministic unless
further constraints (e.g., orthogonality) are imposed. Thus, an input $x$ is ``rejected'' by a rewriting system if{f} there is no reduction to $\symb{true}$ from $\symb{decide}(\encode{x})$\CKchange{;} and as evaluation
is non-deterministic, there may be many distinct reductions starting from $\symb{decide}(\encode{x})$.

\section{Cons-free rewriting}\label{sec:limitations}

Since the purpose of this research is to find groups of programs which
can handle \emph{restricted} classes of Turing-computable problems, we
will impose certain limitations.  In
particular, we will limit interest to \emph{cons-free left-linear
constructor rewriting systems}.

\begin{definition}\label{def:consfree}
A rule $\ell \arrz r$, presented using $\alpha$-conversion in a form
where all binders are distinct from $\FV(\ell)$, is
\emph{cons-free} if for all subterms $s = f(s_1,\dots,s_n) \subtermeq r$
with $f \in \Constructors$, we have $s \subterm \ell$ or $s \in \Data$.
A left-linear constructor AFS $(\F,\Rules)$ is cons-free if all rules
in $\Rules$ are.
\end{definition}

This definition corresponds largely to the definitions of
cons-freeness appearing in~\cite{car:sim:14,jon:01}.  In a cons-free
AFS, it is not possible to create more data, as we will see in
Section~\ref{subsec:properties}.

\begin{example}
The AFS from Example~\ref{ex:plusfo} is not cons-free due to rules
(B) and (C).  The AFS from
Example~\ref{ex:plusho} \emph{is} cons-free (in rules (E) and (G),
$\unknown{a}(xs)$ is allowed to occur on the right despite
the constructor $\unknown{a}$, because it also occurs on the left).
However, there are few interesting basic terms, as we do not consider
for instance $\symb{succ}(\abs{x}{\sfalse},\nil)$ basic.
\end{example}

\begin{remark}
The limitation to left-linear constructor AFSs is
standard, but also \emph{necessary}: if either restriction is
dropped, our limitation to cons-free AFSs becomes
meaningless.  In the case of constructor systems, this is obvious:
if defined symbols are allowed to occur within a left-hand side, then
we could simply let $\Defineds := \F$ and have a ``cons-free'' system.
The case of left-linearity is a bit more sophisticated; this we will
study in more detail in Section~\ref{sec:changes}.

As the first two restrictions are necessary to give meaning to the
third, we will \CKchange{consider the limitation to left-linear
constructor AFSs implicit in the notion ``cons-free''}.
\end{remark}

\subsection{Properties of Cons-free Term Rewriting}
\label{subsec:properties}

As mentioned, cons-free term rewriting cannot create new data.  This
means that the set of data terms that might occur during a reduction
starting in some basic term $s$ are exactly the data terms occurring
in $s$, or those occurring in the right-hand side of some rule.
Formally:

\begin{definition}
Let $(\F,\Rules)$ be a constructor AFS. For a given term
$s$, the set $\B_s$ contains all data terms $t$ such that
(i) $s \suptermeq t$,
or (ii) $r \suptermeq t$ for some rule $\ell \arrz r \in \Rules$.
\end{definition}

$\B_s$ is a set of data terms, is closed under subterms and, since we
have assumed $\Rules$ to be fixed, has a linear number of elements in
the size of
$s$.  The property that no new data is generated by reducing $s$ is
formally expressed by the following result:

\begin{definition}[$\B$-safety]
Let $\B \subseteq \Data$ be a set which (i) is closed under taking
subterms, and (ii) contains all data terms occurring as a subterm of
the right-hand side of a rule in $\Rules$.  A term $s$ is
\emph{$\B$-safe} if for all $t$ with $s \suptermeq t$: if $t$ has the
form $c(t_1,\dots,t_m)$ with $c \in \Constructors$, then $t \in \B$.
\end{definition}

\edef\safetypreservelem{\number\value{theorem}}

\begin{lemma}\label{lem:safetypreserve}
If $s$ is $\B$-safe and $s \arr{\Rules} t$, then $t$ is $\B$-safe.
\end{lemma}

\begin{proof}[Proof Sketch]
By induction on the form of $s$; the result follows trivially by the
induction hypothesis if the reduction does not take place at the root,
leaving only the base cases $s = \app{(\abs{x}{u})}{v} \arr{\Rules}
u[x:=v] = t$ and $s = \ell\gamma \arr{\Rules} r\gamma = t$.  The first
of these is easy by induction on the form of the ($\B$-safe!) term
$u$, the second follows by induction on the form of $r$ (which, as the
right-hand side of a cons-free rule, has convenient properties).
\end{proof}

Thus, if we start with a basic term $f(\vec{s})$,
any data terms occurring in a reduction $f(\vec{s}) \arrr{\Rules}
t$ (directly or as subterms) are in $\B_{f(\vec{s})}$.
This insight will be instrumental in Section~\ref{sec:algorithm}.

\begin{example}\label{ex:palindrome}
By Lemma~\ref{lem:safetypreserve}, functions in a
cons-free AFSs cannot
%use helper functions which
build recursive data.
%Thus, we may need to ``code around'' a problem, for instance using
To code around this, we might use
subterms of the input as a measure of length.  Consider the
decision problem whether a given bitstring is a palindrome.
We cannot use a rule such as $\symb{decide}(cs) \arrz
\symb{equal}(cs,\symb{reverse}(cs))$ since, by
Lemma~\ref{lem:safetypreserve}, it is impossible to define
$\symb{reverse}$.
Instead, a typical solution uses a string $ys$ of length $k$ to find
$\encode{c_k}$ in $\encode{c_0\dots c_{n-1}}$:
\[
\begin{array}{rcl}
\symb{decide}(cs) & \arrz & \symb{palindrome}(cs,cs) \\
\symb{palindrome}(cs, \nil) & \arrz & \strue \\
\symb{palindrome}(cs, \unknown{a}(ys)) & \arrz &
  \symb{and}(\symb{palindrome}(cs,ys),
             \symb{chk}_{\unknown{a}}(cs, ys))\ \ \ \ \ \ \ \ 
  \hfill \llbracket \unknown{a} \in \{\nul,\one\}\rrbracket \\
\end{array}
\]
\vspace{-8pt}
\[
\begin{array}{rclrcl}
\symb{and}(\strue,x) & \arrz & x &
\symb{chk}_{\unknown{a}}(\unknown{a}(xs),\nil) & \arrz & \strue\ \ 
  \hfill \llbracket \unknown{a} \in \{\nul,\one\}\rrbracket \\
\symb{and}(\sfalse,x) & \arrz & \sfalse &
\symb{chk}_{\unknown{a}}(\unknown{b}(xs),\nil) & \arrz & \sfalse
  \ \ \ \hfill \llbracket \unknown{a},\unknown{b} \in
  \{\nul,\one\} \wedge \unknown{a} \neq \unknown{b}\rrbracket \\
& & & 
\symb{chk}_{\unknown{a}}(\unknown{b}(xs),\unknown{c}(ys)) & \arrz &
  \symb{chk}_{\unknown{a}}(xs,ys)\ \ 
  \hfill \llbracket \unknown{a},\unknown{b},\unknown{c} \in \{\nul,
  \one\}\rrbracket \\
\end{array}
\]
(The signature extends $\F_\bits$, but is otherwise omitted as types
can easily be derived.)
\end{example}

Through cons-freeness, we obtain another useful property: we do not
have to consider constructors which take functional
arguments.

\begin{lemma}\label{lem:niceconstructor}
Given a cons-free AFS $(\F,\Rules)$ with $\F = \Defineds \cup
\Constructors$, let $Y = \{ c : [\atype_1 \times
\dots \times \atype_n] \arrtype \asort \in \Constructors$ some
$\atype_i$ is not a sort$\}$.  Define $\F' := \F \setminus Y$, and
let $\Rules'$ consist of those rules in $\Rules$ not using any
element of $Y$ in either left- or right-hand side.  Then
%\begin{itemize}
%\item
  (a) all data and $\B$-safe terms are in $\Terms(\F',\emptyset)$, and
%\item
  (b) if $s$ is a basic term and $s \arrr{\Rules} t$, then
  $t \in \Terms(\F',\V)$ and $s \arrr{\Rules'} t$.
%\end{itemize}
\end{lemma}

\begin{proof}
Since data terms have base type, and the subterms of data
terms are data terms, we have (a).
%; therefore, $\B$-safe terms are in $\Terms(\F',\V)$.
Then, $\B$-safe terms can only
%Since $\B$-safe terms consequently can only
be matched by rules in $\Rules'$,
so Lemma~\ref{lem:safetypreserve} gives (b).
\end{proof}

Therefore we may safely assume that all elements of $\Constructors$
are at most first-order.

\subsection{A larger example}\label{subsec:sat}

None of our examples so far have taken advantage of the
native non-determinism of term rewriting.  To demonstrate the
possibilities, we consider a first-order cons-free AFS that solves the
Boolean satisfiability problem (SAT).  This is striking because, in Jones' language
in~\cite{jon:01}, first-order programs cannot solve this problem
unless P = NP, even if a non-deterministic $\symb{choose}$ operator is
added~\cite{DBLP:conf/amast/Bonfante06}.  The crucial difference is
that we, unlike Jones, do not employ a call-by-value evaluation
strategy.

Given $n$ boolean variables $x_1,\dots,x_n$ and a boolean formula
$\psi ::= \varphi_1 \wedge \dots \wedge \varphi_n$, the satisfiability
problem considers whether there is an assignment of each $x_i$ to
$\top$ or $\bot$ such that $\psi$ evaluates to $\top$.  Here, each
clause $\varphi_i$ has the form $a_{i_1} \vee \dots \vee a_{i_{k_i}}$,
where each literal $a_{i_j}$ is either some $x_p$ or $\neg x_p$.  We  
represent this problem as a string over $I:=\{0,1,\#,?\}$: the formula
$\psi$ is represented by $L::=b_{1,1}\dots b_{1,n}\# b_{2,1}\dots\#
b_{m,1}\dots b_{m,n}\#$, where each $b_{i,j}$ is $1$ if $x_j$ is a
literal in $\varphi_i$, is $0$ if $\neg x_j$ is a literal in $\varphi_i$,
and is $?$ otherwise.

\begin{example}\label{sat:base}
The satisfiability problem for $(x_1 \vee \neg x_2) \wedge (x_2 \vee
\neg x_3)$ is encoded as $\symb{10?\#?10\#}$.
\end{example}

Letting $\nul,\one,\symb{\#},\symb{?} :
[\symb{string}] \arrtype \symb{string}$, and assuming other
declarations clear from context, we claim that the AFS in
Figure~\ref{fig:sat} can reduce $\symb{decide}(\encode{L})$ to
$\symb{true}$ iff $\psi$ is satisfiable.

\begin{figure}[h]
\vspace{-8pt}
\[
\left.
\begin{array}{rclrcl}
\symb{eq}(\symb{\#}(xs),\symb{\#}(ys)) & \arrz & \strue &
\symb{eq}(\symb{\#}(xs),\unknown{a}(ys)) & \arrz & \sfalse \\
\symb{eq}(\unknown{a}(xs),\unknown{b}(ys)) & \arrz & \symb{eq}(xs,ys) &
\symb{eq}(\unknown{a}(xs),\symb{\#}(ys)) & \arrz & \sfalse \\
\end{array}
\right\}\ \llbracket\text{for}\ \unknown{a},\unknown{b} \in \{\nul,
  \one,\symb{?}\}\rrbracket
\]
\vspace{-10pt}
\[
\begin{array}{rcl}
\symb{decide}(cs) & \arrz & \symb{assign}(cs,\nil,\nil,cs) \\
\symb{assign}(\symb{\#}(xs),s,t,cs) & \arrz & \symb{main}(s,t,cs) \\
\end{array}
\]
\vspace{-10pt}
\[
\left.
\begin{array}{rcl}
\,\symb{assign}(\unknown{a}(xs),s,t,cs) & \arrz &
  \symb{assign}(xs,\symb{either}(\unknown{a}(xs),s),t,cs) \\
\symb{assign}(\unknown{a}(xs),s,t,cs) & \arrz &
  \symb{assign}(xs,s,\symb{either}(\unknown{a}(xs),t),cs) \\
\end{array}
\right\}\ \llbracket\text{for}\ \unknown{a} \in \{\nul,\one,\symb{?}\}
  \rrbracket
\]
\vspace{-8pt}
\[
\begin{array}{rclcrcl}
\symb{either}(xs,q) & \arrz & xs & &
\symb{either}(xs,q) & \arrz & q \\
\end{array}
\]
\vspace{-12pt}
\[
\begin{array}{rcl}
\symb{main}(s,t,\symb{?}(xs)) & \arrz & \symb{main}(s,t,xs) \\
\symb{main}(s,t,\nul(xs)) & \arrz &
  \symb{test}(s,t,xs,\symb{eq}(t,\nul(xs)),\symb{eq}(s,\nul(xs))) \\
\symb{main}(s,t,\one(xs)) & \arrz &
  \symb{test}(s,t,xs,\symb{eq}(s,\one(xs)),\symb{eq}(t,\one(xs))) \\
\end{array}
\]
\vspace{-6pt}
\[
\begin{array}{rclcrcl}
\symb{main}(s,t,\nil) & \arrz & \strue & &
\symb{test}(s,t,xs,\strue,z) & \arrz &
  \symb{main}(s,t,\symb{skip}(xs)) \\
\symb{main}(s,t,\symb{\#}(xs)) & \arrz & \sfalse & &
\symb{test}(s,t,xs,z,\strue) & \arrz & \symb{main}(s,t,xs) \\
\end{array}
\]
\vspace{-7pt}
\[
\begin{array}{rcl}
\symb{skip}(\symb{\#}(xs)) & \arrz & xs \\
\symb{skip}(\unknown{a}(xs)) & \arrz & \symb{skip}(xs)\ 
  \llbracket\text{for}\ \unknown{a} \in \{\nul,\one,\symb{?}\}
  \rrbracket \\
\end{array}
\]
\caption{A cons-free first-order AFS solving the satisfiability problem}
\label{fig:sat}
\end{figure}
%CK: I made quite a few changes to the variable namings, to appease
%referee two.  It may be worth checking whether you agree it all
%looks correct, Jakob.

% Yes, it looks fine.

In this AFS, we follow some of the same ideas as in
Example~\ref{ex:palindrome}.  In particular, any string of the form
$b_i\dots b_n\#\dots$ with each $b_j \in \{0,1,?\}$ is considered to
represent the number $i$.  The rules for $\symb{eq}$ are defined so
that $\symb{eq}(s,t)$ tests equality of these \emph{numbers}, not
the full strings.

The key idea new to this example is that we use terms not
in normal form to represent a \emph{set} of numbers.  If we are
interested in %---and can represent---
numbers in $\{1,\dots,n\}$, then
a set $X \subseteq \{1,\dots,n\}$ is encoded as a pair $(s,t)$ of
terms such that, for $i \in \{1,\dots,n\}$:
%\begin{itemize}
%\item
  $s \arrr{\Rules} q$ for some representation $q$ of $i$
  if and only if $i \in X$,
and
%\item
  $t \arrr{\Rules} q$ for some representation $q$ of $i$
  if and only if $i \notin X$;
%\end{itemize}

This is possible because we do not use a call-by-value or
similar reduction strategy: an evaluation of this AFS is allowed to
postpone reducing such terms, and we focus on those
reductions.  The AFS is constructed in such a way that reductions
which evaluate these ``sets'' too eagerly simply end in an
irreducible, non-data state.

Now, an evaluation starting in $\symb{decide}(\encode{L})$
first non-deterministically constructs a ``set''$X$ containing those boolean variables
assigned $\symb{true}$:
$\symb{decide}(\encode{L}) \arrr{\Rules} \symb{main}(s,t,\encode{L})$.
Then, the main function goes through $\encode{L}$, finding for each
clause a literal that is satisfied by the assignment.
%To test whether a literal is satisfied, say encountering
\CKchange{Encountering for instance}
$b_{i_j} =
\symb{1}$, we determine if $j \in X$ by comparing both a
reduct of
$s$ and of $t$ to $j$.  If $s \arrr{\Rules}$ ``$j$'' then $j \in X$,
if $t \arrr{\Rules}$ ``$j$'' then $j \notin X$; in either case, we
continue accordingly.  If the evaluation state is incorrect, or
if $s$ or $t$ both non-deterministically reduce to some other term, the
evaluation \CKchange{gets} stuck in a non-data normal form.

\begin{example}\label{ex:sat}
To solve satisfiability of $(x_1 \vee \neg x_2) \wedge
(x_2 \vee \neg x_3)$, we reduce $\symb{decide}(\encode{L})$, where
$L =\symb{10?\#?10\#}$.  First, we build a valuation; the
choices made by the $\symb{assign}$ rules are non-deterministic, but a
possible reduction is $\symb{decide}(\encode{L}) \arrr{\Rules}
\symb{main}(s,t,\encode{L})$, where $s =
\symb{either}(\encode{10?\#?10\#},\nil)$ and
$t = \symb{either}(\encode{?\#?10\#},\symb{either}(\encode{0?\#?10\#},
  \nil))$.
Recall that, since $n = 3$, $\encode{10?\#?10\#}$ represents $1$ while
$\encode{?\#?10\#}$ and $\encode{0?\#?10\#}$ represent $3$ and $2$
respectively.
Thus, this corresponds to the valuation $[x_1:=\top,x_2:=\bot,x_3:=
\bot]$.

Then, the main loop recurses over the problem.  Note that
$s$ reduces to a term $\encode{10?\#\dots}$ and $t$ reduces to both
$\encode{?\#\dots}$ and $\encode{0?\#\dots}$.  Therefore,
$\symb{main}(s,t,\encode{L}) = \symb{main}(s,t,\encode{11?\#?01\#})
\arrr{\Rules} \symb{main}(s,t,\symb{skip}(\encode{1?\#?01\#}))
\arrr{\Rules} \symb{main}(s,t,\encode{?01\#})$:
the first clause is confirmed since $x_1$ is mapped to $\top$, so the
clause is removed and the loop continues with the second clause.
Next, the loop passes over those variables whose assignment does not
contribute to verifying this clause, until the clause is confirmed by
$x_3$:
$\symb{main}(s,t,\encode{?01\#}) \arr{\Rules}
\symb{main}(s,t,\encode{01\#}) \arrr{\Rules}
\symb{main}(s,t,\encode{1\#}) \arrr{\Rules}
\symb{main}(s,t,\symb{skip}(\encode{\#})) \arr{\Rules}
\symb{main}(s,t,\nil) \arr{\Rules} \strue$.
\end{example}

Using non-determinism, the term in
Example~\ref{ex:sat} could easily have been reduced to $\sfalse$
instead, simply by selecting a different valuation.  This is not
problematic: by definition, the AFS accepts the set of satisfiable
formulas if $\symb{decide}(\encode{L}) \arrr{\Rules} \strue$ if and
only if $L$ is a satisfiable formula%.  That is, we are \emph{only}
%interested in reductions to $\strue$;
\CKchange{:} false negatives or reductions
which do not end in a data state are allowed.

\smallskip
\fscdversion{A longer example derivation is given in Appendix B of the full version of the paper.}
\arxivversion{A longer example derivation is given in Appendix B.}

\section{Simulating $\etime{k}$ Turing machines}\label{sec:counting}

In order to see that cons-free term rewriting captures certain classes
of decidable sets, we will simulate Turing Machines.  For this,
we use an approach very similar to that by Jones~\cite{jon:01}.
We introduce constructor symbols $\symb{a} : [\bits] \arrtype \bits$
for all $a \in A$ (including the blank symbol, which we shall refer
to as $\symb{B}$) along with $\nil$ and the booleans, $\symb{s} :
\symb{state}$ for all $s \in S \cup \{\text{fail}\}$,\ $\symb{L},
\symb{R} : \symb{direction}$ and $\symb{action} : [\bits \times
\symb{direction} \times \symb{state}] \arrtype \symb{trans}$,
$\symb{end} : [\symb{state}] \arrtype \symb{trans}$,
$\symb{NA} : \symb{trans}$.
We will introduce defined symbols and rules such that, for any
string $c \in (A \setminus \{\blank\})^*$---represented as the term
$cs:=c_1(c_2(\cdots c_n(\nil)\cdots))$---we have:
\begin{itemize}
\item $\symb{decide}(cs) \arrr{\Rules} \symb{true}$ if and only if
  $(\blank c\blank\blank\dots,0,\symb{start}) \arrzt^* (t,i,
  \symb{accept})$ for some $t,i$;
\item $\symb{decide}(cs) \arrr{\Rules} \symb{false}$ if and only if
  $(\blank c\blank\blank\dots,0,\symb{start}) \arrzt^* (t,i,
  \symb{reject})$ for some $t,i$.
\end{itemize}
As rules may be overlapping, it is possible that $\symb{decide}(cs)$
will have additional normal forms, but only one normal form will be a
data term.

The rough idea of the simulation is to represent non-negative integers
as terms and let $\symb{tape}(n,p)$ reduce to the symbol at
position $p$ on the tape at the start of the $n^{\text{th}}$ step,
while $\symb{state}(n,p)$ returns the state of the machine
at time $n$, provided the tape head is at position $p$.
If the tape head of the machine is not at position $p$ at time $n$, then
$\symb{state}(n,p)$ should return $\symb{fail}$ instead; this makes
it possible to test the position of the tape head at any given time.
As the machine is deterministic, we can devise rules to compute these
terms from earlier configurations.

Finding a suitable representation of integers and corresponding
manipulating functions is the most intricate part of this simulation, where
we may need both higher-order functions and non-deterministic rules.
Therefore, let us first assume that this can be done.
Then, for a Turing machine which is given to run in time bounded
above by $\lambda x.P(x)$, we define the AFS in Figure~\ref{fig:TM}.
Note that, by construction, any occurrence of $cs$ can only be
instantiated by the input string \CKchange{during evaluation}.

\begin{figure}[htb]
\vspace{-10pt}
\[
\left.
\begin{array}{rcl}
\symb{ifelse}_{\unknown{\asort}}(\symb{true}, y, z) & \arrz & y \\
\symb{ifelse}_{\unknown{\asort}}(\symb{false}, y, z) & \arrz & z \\
\end{array}
\right\}\ 
\llbracket\text{for}\ \unknown{\asort} \in
\{\bits,\symb{state}\}\rrbracket
\]
\vspace{-10pt}
\[
\begin{array}{rcl}
\symb{get}(\nil,\numrep{i},q) & \arrz & q \\
\symb{get}(\unknown{a}(xs),\numrep{i},q) & \arrz &
  \symb{ifelse}_\bits(\numrep{i=0},
  \unknown{a}(\nil),
  \symb{get}(xs,\numrep{i-1},q))\ \ 
  \llbracket\text{for all}\ \unknown{a} \in I\rrbracket
\end{array}
\]
\vspace{-9pt}
\[
\begin{array}{rcl}
\symb{inputtape}(cs,\numrep{p}) & \arrz &
  \symb{ifelse}_{\bits}(\numrep{p=0},
    \symb{B}(\nil),\symb{get}(cs,\numrep{p-1},\symb{B}(\nil))) \\
\end{array}
\]
\vspace{-10pt}
\[
\begin{array}{rcl}
\symb{tape}(cs,\numrep{n},\numrep{p}) & \arrz &
  \symb{ifelse}_{\bits}(\numrep{n=0},
  \symb{inputtape}(cs,\numrep{p}),\symb{tapex}(cs,
  \numrep{n-1},\numrep{p})) \\
\symb{tapex}(cs,\numrep{n},\numrep{p}) & \arrz &
  \symb{tapey}(cs,\numrep{n},\numrep{p},
  \symb{transition}(cs,\numrep{n},\numrep{p})) \\
\end{array}
\]
\vspace{-8pt}
\[
\begin{array}{rclrcl}
\symb{tapey}(cs,\numrep{n},\numrep{p},\symb{action}(q,d,s)) & \arrz &
  q &
\symb{tapey}(cs,\numrep{n},\numrep{p},\symb{NA}) & \arrz &
  \symb{tape}(cs,\numrep{n},\numrep{p}) \\
& & &
\!\!\!\!\!\!\!\!
\symb{tapey}(cs,\numrep{n},\numrep{p},\symb{end}(s)) & \arrz &
  \symb{tape}(cs,\numrep{n},\numrep{p}) \\
\end{array}
\]
\vspace{-5pt}
\[
\begin{array}{rcl}
\symb{state}(cs,\numrep{n},\numrep{p}) & \arrz &
  \symb{ifelse}_{\symb{state}}(\numrep{n=0},\symb{state0}(cs,\numrep{p}),
  \symb{statex}(cs,\numrep{n-1},\numrep{p})) \\
\symb{state0}(cs,\numrep{p}) & \arrz &
  \symb{ifelse}_{\symb{state}}(\numrep{p=0},\symb{start},\symb{fail}) \\
\symb{statex}(cs,\numrep{n},\numrep{p}) & \arrz &
  \symb{statey}(\symb{transition}(cs,\numrep{n},\numrep{p-1}),
  \symb{transition}(cs,\numrep{n},\numrep{p}), \\
  & & \phantom{\symb{statey}(}\symb{transition}(cs,\numrep{n},
  \numrep{p+1})) \\
\end{array}
\]
\vspace{-8pt}
\[
\begin{array}{rclrcl}
\symb{statey}(\symb{action}(q,\symb{R},s),y,z) & \arrz & s &
\symb{statey}(\symb{NA},\symb{action}(q,d,s),z) & \arrz & 
  \symb{fail} \\
\symb{statey}(\symb{action}(q,\symb{L},s),y,z) & \arrz & \symb{fail} &
\symb{statey}(\symb{NA},\symb{NA},\symb{action}(q,\symb{L},s)) &
  \arrz & s \\
\symb{statey}(\symb{end}(s),y,z) & \arrz & \symb{fail}\ \ \ \  &
\symb{statey}(\symb{NA},\symb{NA},\symb{action}(q,\symb{R},s)) &
  \arrz & \symb{fail} \\
\symb{statey}(\symb{NA},\symb{end}(s),z) & \arrz & s &
\symb{statey}(\symb{NA},\symb{NA},\symb{end}(s)) &
  \arrz & \symb{fail} \\
\end{array}
\]
\vspace{-4pt}
\[
\begin{array}{rcl}
\symb{transition}(cs,\numrep{n},\numrep{p}) & \arrz &
  \symb{transitionhelp}(\symb{state}(cs,\numrep{n},\numrep{p}),
                        \symb{tape}(cs,\numrep{n},\numrep{p})) \\
\vspace{-5pt}
\symb{transitionhelp}(\symb{fail},q) & \arrz & \symb{NA} \\
\symb{transitionhelp}(\unknown{s},\unknown{r}(\nil)) & \arrz &
  \symb{action}(\unknown{w}(\nil),\unknown{d},\unknown{t})\hfill
  \llbracket\text{for all}\ \transition{\unknown{s}}{\unknown{r}}{
  \unknown{w}}{\unknown{d}}{\unknown{t}} \in T\rrbracket \\
\symb{transitionhelp}(\unknown{s},q) & \arrz & \symb{end}(
  \unknown{s}) \hfill\llbracket\text{for}\ \unknown{s} \in
  \{\symb{accept},\symb{reject}\}\rrbracket
\end{array}
\]
\vspace{-3pt}
\[
\begin{array}{rcl}
\symb{decide}(cs) & \arrz & \symb{findanswer}(cs,\numrep{P(|cs|)},
  \numrep{P(|cs|)}) \\
\symb{findanswer}(cs,\numrep{n},\numrep{p}) & \arrz &
  \symb{test}(\symb{state}(cs,\numrep{n},\numrep{p}),
  cs,\numrep{n},\numrep{p}) \\
\symb{test}(\symb{fail},cs,\numrep{n},\numrep{p}) & \arrz &
  \symb{findanswer}(cs,\numrep{n},\numrep{p-1}) \\
\symb{test}(\symb{accept},cs,\numrep{n},\numrep{p}) & \arrz & \strue \\
\symb{test}(\symb{reject},cs,\numrep{n},\numrep{p}) & \arrz & \sfalse \\
\end{array}
\]
\caption{Simulating a deterministic Turing Machine running in
$\lambda x.P(x)$ time.}
\label{fig:TM}
\end{figure}

\paragraph*{Counting}

The goal, then, is to find a representation of numbers and
functionality to do four things:
\begin{itemize}
\item calculate $\numrep{P(|cs|)}$ or an overestimation (as
  the machine cannot move from its final state);
\item test whether a ``number'' represents $0$;
\item given $\numrep{n}$, calculate $\numrep{n-1}$,
  \emph{provided $n > 0$}---so it suffices to determine
  $\numrep{\max(n-1,0)}$;
\item given $\numrep{n}$, calculate $\numrep{n+1}$,
  \emph{provided $n+1 \leq P(|cs|)$} as necessarily
  $\symb{transition}(cs,\numrep{n},\numrep{p})\linebreak
  \arr{\Rules} \symb{NA}$ when $n < p$ and $n$ never increases---so
  it suffices to determine $\numrep{\min(n+1,P(|cs|))}$.
\end{itemize}
Moreover, these calculations all occur in the right-hand side of a
rule containing the initial input list $cs$ on the left,
which they can therefore use (for instance to recompute $P(|cs|)$).

Rather than representing a number by a single term, we will use
\emph{tuples} of terms
(which are not terms themselves, as a pairing constructor would
conflict with cons-freeness).
To illustrate %how this works,
this,
suppose we represent each number $n$ by
a pair $(n_1,n_2)$.
Then the predecessor and successor function must also be split,
e.g.~$\symb{pred}^1(cs,n_1,n_2) \arrr{\Rules} n_1'$ and
$\symb{pred}^2(cs,n_1,n_2) \arrr{\Rules} n_2'$ for $(n_1',n_2')$ some
tuple representing $n-1$.  Thus, for instance the first
$\symb{test}$ rule becomes:
\[
\symb{test}(\symb{fail},cs,n_1,n_2,p_1,p_2) \arrz
  \symb{findanswer}(cs,n_1,n_2,\symb{pred}^1(cs,p_1,p_2),
  \symb{pred}^2(cs,p_1,p_2))
\]

Following Jones~\cite{jon:01}, we use the notion of a \emph{counting
module} which provides an AFS with a representation of a counting
function and a means of computing.  Counting  modules can be composed,
making it possible to count to greater numbers. Due to the laxity of
term rewriting, our constructions are technically quite different
from those of~\cite{jon:01}.

\begin{definition}[Counting Module]
Write $\F = \Constructors \cup \Defineds$ for the signature in
Figure~\ref{fig:TM}.  For $P$ a function from $\nats$ to $\nats$, a
$P$-counting module of \emph{order $k$} is a tuple
$C_{\pi} ::= (\vec{\atype},\Sigma,R,A,\numinterpret{\cdot})$ s.t.:
\begin{itemize}
\item $\vec{\atype}$ is a sequence of types $\atype_1 \times \dots
  \times \atype_a$ where each $\atype_i$ has order at most $k-1$;
\item $\Sigma$ is a $k^{\text{th}}$-order signature disjoint from
  $\F$, with designated symbols $\symb{zero}_{\pi} : [\bits \times
  \vec{\atype}]
  \arrtype \bool$ and, for $1 \leq i \leq a$ with $\atype_i =
  \btype_1 \arrtype \dots \arrtype \btype_m \arrtype \asort$
  symbols $\symb{pred}_{\pi}^i,\symb{suc}_{\pi}^i,
  \symb{inv}_{\pi}^i : [\bits \times \vec{\atype} \times
  \vec{\btype}] \arrtype \asort$ and $\symb{seed}_{\pi}^i :
  [\bits \times \vec{\btype}] \arrtype \bsort$;
\item $R$ is a set of cons-free rules $f(\vec{\ell}) \arrz
  r$ with
  $f \in \Sigma$, each $\ell_i \in \Terms(\Constructors,\V)$ and $r
  \in \Terms(\Constructors \cup \Sigma,\V)$;
\item for every string $cs \subseteq I^+$, the set $A_{cs} \subseteq
  \{ (s_1,\dots,s_a) \in \Terms(\Constructors \cup \Sigma)^a \mid
  \vdash s_j : \atype_j$  for $1 \leq j \leq a\}$;
\item for every string $cs$, $\numinterpret{\cdot}_{cs}$
  is a surjective mapping from $A_{cs}$ to $\{0,\dots,P(|cs|)-1\}$;
\item writing e.g.~$\symb{pred}^i_{\pi}[\vec{s}] : \atype_i$ for the term
  $\abs{\vec{y}}{\symb{pred}^i_{\pi}(\vec{s},\vec{y})}$, the following
  properties are satisfied:
  \begin{itemize}
  \item $(\symb{seed}^1_{\pi}[cs],\dots,\symb{seed}^a_{\pi}[cs]) \in A_{cs}$
    and $\numinterpret{(\symb{seed}^1_{\pi}[cs],\dots,\symb{seed}^a_{\pi}[cs]
    )}_{cs} = P(|cs|)-1$
  \end{itemize}
  and for all $(s_1,\dots,s_a) \in A_{cs}$ with $\numinterpret{(s_1,
  \dots,s_a)}_{cs} = m$:
  \begin{itemize}
  \item
    $(\symb{pred}^1_{\pi}[cs,\vec{s}],\dots,\symb{pred}^a_{\pi}[cs,\vec{s}])$
    and
    $(\symb{suc}^1_{\pi}[cs,\vec{s}],\dots,\symb{suc}^a_{\pi}[cs,\vec{s}])$
    and
    $(\symb{inv}^1_{\pi}[cs,\vec{s}],\dots,\linebreak
    \symb{inv}^a_{\pi}[cs,\vec{s}])$
    are all in $A_{cs}$
  \item
    $\numinterpret{(\symb{pred}^1_{\pi}[cs,\vec{s}],\dots,
    \symb{pred}^a_{\pi}[cs,\vec{s}])}_{cs} = \max(m-1,0)$
  \item
    $\numinterpret{(\symb{suc}^1_{\pi}[cs,\vec{s}],\dots,
    \symb{suc}^a_{\pi}[cs,\vec{s}])}_{cs} = \min(m+1,P(|cs|)-1)$
  \item
    $\numinterpret{(\symb{inv}^1_{\pi}[cs,\vec{s}],\dots,
    \symb{inv}^a_{\pi}[cs,\vec{s}])}_{cs} = P(|cs|)-1-m$
  \item $\symb{zero}_{\pi}(cs,\vec{s}) \arrr{R} \strue$ if{f} $m = 0$ and
    $\symb{zero}_{\pi}(cs,\vec{s}) \arrr{R} \sfalse$ if{f} $m > 0$
  \item if each $s_i \arrr{R} t_i$ and $(t_1,\dots,t_a) \in A_{cs}$,
    then also $\numinterpret{(t_1,\dots,t_a)}_{cs} = m$.
  \end{itemize}
\end{itemize}
\end{definition}

It is not hard to see how we would use a $P$-counting module in the
AFS of Figure~\ref{fig:TM}; this results in a $k^{\text{th}}$-order
AFS for a $k^{\text{th}}$-order module.  Note that this works even if
some number representations $(s_1,\dots,s_a)$ are not in normal form:
even if we reduce $\vec{s}$ to some tuple $\vec{t}$, the result of the
$\symb{zero}$ test cannot change from $\strue$ to $\sfalse$ or vice
versa.  Since the algorithm relies heavily on these tests, we may
safely assume that terms representing numbers are reduced in a lazy
way---as we did in Section~\ref{subsec:sat} for the arguments $s$ and
$t$ of $\symb{main}$.

\edef\basemodule{\number\value{theorem}}
\begin{lemma}\label{lem:mainmodule}
There is a first-order ($\lambda n.2^{n+1}$)-counting module.
\end{lemma}

\begin{proof}
%We reuse the ideas of
\CKchange{As in}
Section~\ref{subsec:sat},
\CKchange{we will represent}
%representing
a set of numbers---or rather, its encoding as a
bit-sequence---by a pair of terms.
We let $C_{\symb{e}} := (\symb{string} \times \symb{string},\Sigma,R,
A,\numinterpret{\cdot})$, where:
\begin{itemize}
\item $A_{cs}$ contains all pairs $(s,t)$ such that (a) all data terms
  $q$ such that $s \arrr{R} q$ or $t \arrr{R} q$ are subterms of $cs$,
  and (b) for each $q \subtermeq cs$ either $s \arrr{R} q$ or $t
  \arrr{R} q$, but not both.
\item Writing $cs = c_N(\dots(c_1(\nil))\dots)$, we let $cs_0 =
  \nil,\ cs_1 = c_1(\nil)$ and so on.  We let
  $\numinterpret{(s,t)}_{cs} =
  \sum_{i = 0}^N \{ 2^{N-i} \mid s \arrr{R} cs_i \}$.  That is,
  $\numinterpret{(s,t)}_{cs}$ is the number represented by the
  bit-sequence $b_0\dots b_N$ where $b_i = 1$ if{f} $s \arrr{R}
  cs_i$, if{f} not $t \arrr{R} cs_i$ (with $b_N$ the least
  significant digit).
\item $\Sigma$ consists of the defined symbols introduced in $R$,
  which we construct below.
\end{itemize}
As in Section~\ref{subsec:sat}, we use non-deterministic selection
functions to construct $(s,t)$:
\[
\begin{array}{rclcrclcrcl}
\symb{either}(x,y) & \arrz & x & & \symb{either}(x,y) & \arrz & y & &
\bot & \arrz & \bot \\
\end{array}
\]
The symbol $\bot$ will be used for terms which do not reduce to
any data (the $\bot \arrz \bot$ rule is used to force $\bot \in
\Defineds$).
For the remaining functions, we consider bitvector arithmetic.
First, $2^{N+1}-1$
corresponds to the bit-sequence where each $b_i = 1$:
\[
\begin{array}{rclrcl}
\symb{seed}_{\symb{e}}^1(cs) & \arrz & \symb{all}(cs,\bot) &
\symb{all}(\nil,q) & \arrz & \symb{either}(\nil,q) \\
\symb{seed}_{\symb{e}}^2(cs) & \arrz & \bot &
\symb{all}(\unknown{a}(xs),q) & \arrz &
  \symb{all}(xs,\symb{either}(\unknown{a}(xs),q))\ 
  \llbracket\text{for}\ \unknown{a} \in \symb{I}\rrbracket \\
\end{array}
\]
Here, $\symb{I} = \{ \symb{a} \mid a \in I \}$.  The
inverse function is obtained by flipping the sequence's bits:
\[
\begin{array}{rclcrcl}
\symb{inv}_{\symb{e}}^1(cs,s,t) & \arrz & t & &
\symb{inv}_{\symb{e}}^1(cs,s,t) & \arrz & s \\
\end{array}
\]
In order to define $\symb{zero}_{\symb{e}}$, we must test the value
of all bits in the sequence. This is done by forcing an evaluation
from $s$ or $t$ to some data term. This test is constructed in such
a way that both $\strue$ and $\sfalse$ results necessarily reflect the
state of $s$ and $t$; any undesirable non-deterministic choices lead
to the evaluation getting stuck.
\[
\left.
\begin{array}{rclrcl}
\symb{eqLen}(\nil,\nil) & \arrz & \strue &
\symb{eqLen}(\nil,\unknown{a}(ys)) & \arrz & \sfalse \\
\symb{eqLen}(\unknown{a}(xs),\unknown{b}(ys)) & \arrz &
  \symb{eqLen}(xs,ys) &
\symb{eqLen}(\unknown{a}(xs),\nil) & \arrz & \sfalse \\
\end{array}
\right\}\ \llbracket\text{for}\ \unknown{a},\unknown{b} \in I
\rrbracket
\]
\vspace{-12pt}
\[
\begin{array}{rclcrcl}
\symb{bitset}(xs,s,t) & \arrz & \symb{test}(\symb{eqLen}(xs,s),
  \symb{eqLen}(xs,t)) & &
\symb{test}(\strue,x) & \arrz & \strue \\
& & & & \symb{test}(x,\strue) & \arrz & \sfalse \\
\end{array}
\]
Then $\symb{zero}_{\symb{e}}$ simply tests whether the bit is unset
for each sublist.
\[
\begin{array}{rclrcl}
\symb{zero}_{\symb{e}}(xs,s,t) & \arrz &
  \symb{zo}(xs,s,t,\symb{bitset}(xs,s,t)) &
\symb{zo}(xs,s,t,\strue) & \arrz & \sfalse \\
\symb{zo}(\unknown{a}(xs),s,t,\sfalse) & \arrz &
  \symb{zero}_{\symb{e}}(xs,s,t)\ \llbracket\text{for}\ 
  \unknown{a} \in I\rrbracket &
\symb{zo}(\nil,s,t,\sfalse) & \arrz & \strue \\
\end{array}
\]
For the predecessor function, note that the predecessor of a
bit-sequence $b_0 \dots b_{i-1} b 1 0 \dots 0$ is $b_0 \dots b_{i-1}
0 1 \dots 1$.  We first define a helper function $\symb{copy}$ to copy
$b_0 \dots b_{i-1}$:
\[
\begin{array}{rclcrcl}
\symb{copy}(xs,s,t,\sfalse) & \arrz &
    \multicolumn{5}{l}{\symb{maybeadd}(xs,\symb{bitset}(xs,s,t),
    \symb{copy}(\symb{tl}(xs),s,t,\symb{empty}(xs)))} \\
\symb{copy}(xs,s,t,\strue) & \arrz & \bot & \phantom{XXX} &
\symb{maybeadd}(xs,\strue,q) & \arrz & \symb{either}(xs,q) \\
& & & & \symb{maybeadd}(xs,\sfalse,q) & \arrz & q \\
\end{array}
\]
\vspace{-12pt}
\[
\begin{array}{rclcrclcrcl}
\symb{empty}(\nil) & \arrz & \strue & &
\symb{tl}(\nil) & \arrz & \nil \\
\symb{empty}(\unknown{a}(x)) & \arrz & \sfalse\ \ 
  \llbracket\text{for}\ \unknown{a} \in \symb{I}\rrbracket & & 
\symb{tl}(\unknown{a}(x)) & \arrz & x\ \ \llbracket\text{for}\ 
  \unknown{a} \in \symb{I}\rrbracket \\
\end{array}
\]
Then $\symb{copy}(xs_{\max(i-1,0)},s,t,[i=0])$ reduces to those
$xs_j$ with $0 \leq j < i$ where $b_j = 1$, and $\symb{copy}(xs_{
\max(i-1,0)},t,s,[i=0])$ to those with $b_j = 0$.
This works because $s$ and $t$ are each other's complement.
To define $\symb{pred}$, we first handle the zero case:
\[
\begin{array}{rcl}
\symb{pred}_{\symb{e}}^{\unknown{i}}(cs,s,t) & \arrz &
  \symb{pz}^{\unknown{i}}(cs,s,t,\symb{zero}_{\symb{e}}(cs,s,t))\ 
  \ \llbracket\text{for}\ \unknown{i} \in \{1,2\}\rrbracket \\
\end{array}
\]
\vspace{-18pt}
\[
\begin{array}{rclcrcl}
\symb{pz}^1(cs,s,t,\strue) & \arrz & s & &
\symb{pz}^1(cs,s,t,\sfalse) & \arrz & \symb{pmain}^1(cs,s,t,
  \symb{bitset}(cs,s,t)) \\
\symb{pz}^2(cs,s,t,\strue) & \arrz & t & &
\symb{pz}^2(cs,s,t,\sfalse) & \arrz & \symb{pmain}^2(cs,s,t,
  \symb{bitset}(cs,s,t)) \\
\end{array}
\]
Then, $\symb{pmain}(xs_N,s,t,[b_N = 1])$ flips the bits $b_N,b_{N-1},
\dots$ until an index is encountered where $b_i = 1$; this last bit
is flipped, and the remaining bits copied.  Formally:
\[
\begin{array}{rcl}
\symb{pmain}^1(xs,s,t,\strue) & \arrz &
  \symb{copy}(\symb{tl}(xs),s,t,\symb{empty}(xs)) \\
\symb{pmain}^2(xs,s,t,\strue) & \arrz &
  \symb{either}(xs,\symb{copy}(\symb{tl}(xs),t,s,
  \symb{empty}(xs))) \\
\symb{pmain}^1(xs,s,t,\sfalse) & \arrz &
  \symb{either}(xs,\symb{pmain}^1(\symb{tl}(xs),s,t,
  \symb{bitset}(\symb{tl}(xs),s,t))) \\
\symb{pmain}^2(xs,s,t,\sfalse) & \arrz &
  \symb{pmain}^2(\symb{tl}(xs),s,t,\symb{bitset}(\symb{tl}(xs),s,t)) \\
\end{array}
\]
Finally, we observe that $x + 1 = N - ((N-x)-1)$ and for $x = N$ also
$\min(x+1,N) = N - (\max((N-x)-1,0))$.  Thus, we may define
$\symb{suc}(b)$ as $\symb{inv}(\symb{pred}(\symb{inv}(x)))$.  Taking
pairing into account and writing out the definition, this simplifies
to:
\[
\begin{array}{rclcrcl}
\symb{suc}^1(cs,s,t) & \arrz & \symb{pred}^2(cs,t,s) & &
\symb{suc}^2(cs,s,t) & \arrz & \symb{pred}^1(cs,t,s) \\
\end{array}
\qedhere
\]
\end{proof}

Having Lemma~\ref{lem:mainmodule} as a basis, we can define composite
modules.  Here, we give fewer details than for
Lemma~\ref{lem:mainmodule} as the constructions use many of the same
ideas.

\edef\prodmodule{\number\value{theorem}}
\begin{lemma}\label{lem:RQcount}
If there exist a $P$-counting module $C_{\pi}$ and a $Q$-counting module
$C_{\rho}$, both of order at most $k$, then there is a $(\lambda n.P(n)
\cdot Q(n))$-counting module $C_{\pi\cdot\rho}$ of order at most $k$.
\end{lemma}

\begin{proof}[Proof Sketch]
Let $C_\pi ::= ([\atype_1 \times \dots \times \atype_a],
\Sigma^\pi,R^\pi,A^\pi,\numinterpret{\cdot}^\pi)$ and $C_\rho ::=
([\btype_1 \times \dots \times \btype_b],\Sigma^\rho,R^\rho,\linebreak
A^\rho,
\numinterpret{\cdot}^\rho)$.
We will, essentially, represent the numbers $i \in \{0,\dots,
P(|cs|) \cdot Q(|cs|) - 1\}$ by a pair $(i_1,i_2)$ with $0
\leq i_1 < P(|cs|)$ and $0 \leq i_2 < Q(|cs|)$, such that
$i = i_1 \cdot Q(|cs|) + i_2$.  This is done by defining
$A^{\pi \cdot \rho}_{cs} = \{ (u_1,\dots,u_a,v_1,\dots,v_b) \mid (u_1,\dots,
u_a) \in A^{\pi}_{cs} \wedge (v_1,\dots,v_b) \in A^{\rho}_{cs} \}$, and
$\numinterpret{(\vec{u},\vec{v})}^{\pi \cdot \rho}_{cs} =
\numinterpret{(\vec{u})}^\pi_{cs} \cdot Q(|cs|) +
\numinterpret{(\vec{v})}^\rho_{cs}$.  The signature of
defined symbols and rules of $C_{\pi \cdot \rho}$ are straightforwardly
defined as well, extending those in $C_{\pi}$ and $C_{\rho}$; for instance:
\[
\begin{array}{rcl}
\symb{zero}_{\pi \cdot \rho}(cs,u_1,\dots,u_a,v_1,\dots,v_b) & \arrz &
  \symb{and}(\symb{zero}_{\pi}(cs,u_1,\dots,u_a),\symb{zero}_{\rho}(cs,
  v_1,\dots,v_b))
\end{array}
\]
\vspace{-18pt}
\[
\begin{array}{rclcrcl}
\symb{and}(\strue,x) & \arrz & x & &
\symb{and}(\sfalse,y) & \arrz & \sfalse
\phantom{BLABLABLABLABLABLABL} % TO FIX QED PLACEMENT
\qedhere
\end{array}
\]
\end{proof}

\edef\expmodule{\number\value{theorem}}
\begin{lemma}\label{lem:expQcount}
If there is a $P$-counting module $C_{\pi}$ of order $k$, then there
is a $(\lambda n.2^{P(n)})$-counting module $C_{p[\pi]}$ of order $k+1$.
\end{lemma}

\begin{proof}[Proof Sketch]
We represent
every bitstring $b_{P(|cs|)-1 \cdots b_0}$ as a function of type
$\atype_1 \arrtype \dots \arrtype \atype_a \arrtype \bool$.
The various functions are defined as bitvector operations.  For
example:
\vspace{-2pt}
\[
\begin{array}{rclcrcl}
\symb{seed}_{\symb{p}[\pi]}(cs,k_1,\dots,k_a) & \arrz & \symb{true} & &
\symb{inv}_{\symb{p}[\pi]}(cs,F,k_1,\dots,k_a) & \arrz &
  \symb{not}(\apps{F}{k_1}{k_a}) \\
\end{array}
\]
\vspace{-14pt}
\[
\begin{array}{rcl}
\symb{zero}_{\symb{p}[\pi]}(cs,F) & \arrz & \symb{zero'}_{\symb{p}[\pi]}(cs,
  \symb{seed}^1_{\pi}[cs],\dots,\symb{seed}^a_{\pi}[cs],F) \\
\symb{zero'}_{\symb{p}[\pi]}(cs,k_1,\dots,k_a,F) & \arrz &
  \symb{ztest}_{\symb{p}[\pi]}(\apps{F}{k_1}{k_a},\symb{zero}_{\pi}(cs,k_1,\dots,k_a),cs,
  \\
  & & \phantom{ztest_{p[\pi]}}\ \ \ k_1,\dots,k_a,F) \\
\symb{ztest}_{\symb{p}[\pi]}(\strue,z,cs,\vec{k},F) & \arrz &
  \sfalse \\
\symb{ztest}_{\symb{p}[\pi]}(\sfalse,\strue,cs,\vec{k},F) & \arrz &
  \strue \\
\symb{ztest}_{\symb{p}[\pi]}(\sfalse,\sfalse,cs,\vec{k},F) & \arrz &
  \symb{zero'}_{\symb{p}[\pi]}(cs,\symb{pred}^1_{\pi}[cs,\vec{k}],\dots,
  \symb{pred}^a_{\pi}[cs,\vec{k}], F)
\hfill\qedhere
\end{array}
\]
Note that, for instance, $\symb{seed}_{\symb{p}[\pi]}[cs]$ is
$\abs{k_1 \dots k_a}{\symb{seed}_{\symb{p}[\pi]}(cs,k_1,\dots,k_a)}$:
the additional parameters $k_i$ should be seen as indexing the result
of the function.
\end{proof}

We obtain:

\begin{theorem}\label{thm:simulation}
Any decision problem in $\etime{k}$ can be accepted by a
$k^{\text{th}}$-order AFS.
\end{theorem}

\begin{proof}
Following the construction in this section, it suffices if we can
find a $k^{\text{th}}$-order counting module counting up to
$\mathrm{exp}_2^k(a\cdot n)$ where $n$ is the size of the input and
$a$ a fixed positive integer.  Lemma~\ref{lem:mainmodule} gives a
first-order $\lambda n.2^{n+1}$-counting module, and by iteratively
using Lemma~\ref{lem:RQcount} we obtain $\lambda n.(2^{n+1})^a =
\lambda n.2^{a(n+1)}$ for any $a$.  Iteratively applying
Lemma~\ref{lem:expQcount} on the result gives a
\pagebreak % NOTE: on my machine, this pagebreak makes no
%difference.  I added it because arxiv seems to use a different
%version of pdflatex, and breaks the page before the proof without
%this pagebreak.
$k^{\text{th}}$-order $\lambda n.\mathrm{exp}_2^k(a \cdot (n+1))$-%
counting module.
\end{proof}

\section{Finding normal forms}\label{sec:algorithm}

In the previous section we have seen that every function in
$\etime{k}$ can be implemented by a cons-free $k^{\text{th}}$-order
AFS.  Towards a characterization result, we must therefore show the
converse: that every function implemented by a cons-free
$k^{\text{th}}$-order AFS is in $\etime{k}$.

To achieve this goal, we will now give an algorithm that, on
input any basic term in an AFS of order $k$, will output its set of
data normal forms in $\etime{k}$ in the size of the term.

A key idea is to associate terms of higher-order type to
functions.  We define:
\[
\begin{array}{rcl}
\interpret{\asort} & = & \P(\{ s \mid s \in \B\ \wedge
  \vdash s : \asort \})\ \ \text{for}\ \asort \in \Sorts\ \ 
  (\text{so a set of subsets of}\ \B) \\
\interpret{\atype \arrtype \btype} & = & 
  \interpret{\btype}^{\interpret{\atype}}\ \ 
  (\text{so the set of functions from}\ \interpret{\atype}\ 
  \text{to}\ \interpret{\btype}) \\
\end{array}
\]

Intuitively, an element of $\interpret{\asort}$ represents a set of
possible reducts of a term $s : \asort$, while an element of
$\interpret{\atype \arrtype \btype}$ represents the function defined
by some $\abs{x}{s} : \atype \arrtype \btype$.
Since---as induction on the structure of $\atype$ shows---each
$\interpret{\atype}$ is \emph{finite}, we can define the following
algorithm to find all
normal forms of a given basic term.  In the algorithm, we build
functions $\Conf^0,\Conf^1,\dots$, each mapping statements
$f(A_1,\dots,A_n) \approx t$
to a value in $\{\top,\bot\}$.
Intuitively, $\Conf^i[f(\vec{A}) \approx t]$ denotes whether, in
step $i$ in the algorithm, we have confirmed that $f(s_1,\dots,s_n)
\arrr{\Rules} t$, where each $A_i$ represents the corresponding
$s_i$.

\medskip\noindent\hrulefill\vspace{-8pt}
\begin{algorithm}\label{alg:main}
\quad

{\bf Input:} A basic term $s = g(t_1,\dots,t_m)$.

{\bf Output:} The set of data normal forms of $s$. Note that
this set may be empty.

\smallskip
Set $\B := \B_s$.
For all $f : [\atype_1 \times \dots \times \atype_n] \arrtype
\asort \in \Defineds$, all $A_1 \in \interpret{\atype_1},\dots,
A_n \in \interpret{\atype_n}$, all $t \in \interpret{\asort}$,
we let
$\Conf^0[f(A_1,\dots,A_n) \approx t] := \bot$.
Now, for all such $f,\vec{A},t$ and all $i \in \nats$:
\begin{itemize}
\item if $\Conf^i[f(\vec{A}) \approx t] = \top$, then
  $\Conf^{i+1}[f(\vec{A}) \approx t] := \top$;
\item otherwise,
    for all rules $f(\ell_1,\dots,\ell_n) \arrz r \in \Rules$,
    for all substitutions $\gamma$ on domain  $\FV(f(\vec{\ell
    })) \setminus \{ \vec{\ell} \}$ (so on those variables occurring
    below constructors) such that $\ell_j\gamma \in A_j$ for all $j$
    with $\ell_j$ not a variable ($A_j$ is a set of terms since
    $\ell_j$, a non-variable proper constructor term, must have
    base type),
    let $\eta$ be the function such that for each $\ell_j \in \V$, $\eta(\ell_j) = A_j$,
and
    test whether $t \in \NF^i(r\gamma,\eta)$.
  If there are a rule and substitution where this test succeeds, let
  $\Conf^{i+1}[f(\vec{A}) \approx t] := \top$, otherwise let
  $\Conf^{i+1}[f(\vec{A}) \approx t] := \bot$.
\end{itemize}
Here, $\NF^i(s,\eta)$ is defined recursively for $\B$-safe terms
$s$ and functions $\eta$ mapping all variables $x : \atype$ in
$\FV(s)$ to an element of $\interpret{\atype}$, as follows:
\begin{itemize}
\item if $s$ is a data term, then $\NF^i(s,\eta) :=
  \{ s \}$;
\item if $s$ is a variable, then $\NF^i(s,\eta) := \eta(s)$;
\item if $s = f(s_1,\dots,s_n)$ with $f \in \Defineds$, then
  $\NF^i(s,\eta)$ is the set of all $t \in \B$ such that
  $\Conf^i[f(\NF^i(s_1,\eta),\dots,\NF^i(s_n,\eta)) \approx t] =
  \top$;
\item if $s = \app{u}{v}$, then $\NF^i(s,\eta) = \NF^i(u,\eta)(
  \NF^i(v,\eta))$;
\item if $s =_\alpha \abs{x}{t} : \atype \arrtype \btype$ where $x
  \notin \domain(\eta)$, then $\NF^i(s,\eta) :=$ the function
  mapping $A \in \interpret{\atype}$ to $\NF^i(t,\eta \cup [x:=A])$.
\end{itemize}
When $\Conf^{i+1}[f(\vec{A}) \approx t] = \Conf^i[f(\vec{A}) \approx
t]$ for all
statements,
the algorithm ends; we let $I
:= i+1$ and return $\{ t \in \B \mid \Conf^I[g(\{t_1\},\dots,\{t_m\})
\approx t] = \top\}$.

\noindent
\hrulefill
\end{algorithm}

As $\Defineds,\ \B$ and all $\interpret{\atype_i}$ are all finite,
and the number of positions at which $\Conf^i$ is $\top$ increases
in every step, the algorithm always terminates.
The intention
is that $\Conf^I$ reflects rewriting for basic terms.  This result is
stated formally in
Theorem~\ref{thm:algorithmsoundcomplete}.

\begin{example}
Consider the palindrome AFS in Example~\ref{ex:palindrome}, with
starting term $s = \one(\nul(\nil))$.  Then $\B_s = \{
\one(\nul(\nil)),\nul(\nil),\nil,\strue,\sfalse \}$.  Then we have
$\interpret{\bool} = \{ \emptyset,\{\strue\},\{\sfalse\},\linebreak
\{\strue,
\sfalse\} \}$ and $\interpret{\bits}$ is the set containing all eight
subsets of $\{\one(\nul(\nil)),\nul(\nil),\nil\}$.  Thus, there are
$8 \cdot 8 \cdot 2$ statements of the form $\symb{palindrome}(A,B)
\approx t$, $4 \cdot 4 \cdot 2$ of the form $\symb{and}(A,B) \approx
t$ and so on, totalling $432$ statements to be considered in every
step.

We consider one step, determining
%$\Conf^1[\symb{palindrome}(\{\one(\nul(\nil))\},\{\nul(\nil),\nil\})
%\approx \sfalse]$.
$\Conf^1[\symb{chk}_\one(\{\one(\nul(\nil))\},\{\nul(\nil),\nil\})
\approx \strue]$.
There are two viable combinations of a rule and a substitution:
%$\symb{palindrome}(cs,\nul(ys)) \arrz \symb{and}(
%  \symb{palindrome}(cs,ys),\symb{chk}_\nul(cs,ys))$
$\symb{chk}_\one(\one(xs),\nul(ys)) \arrz \symb{chk}_\one(xs,ys)$
with substitution
$\gamma = [xs:=\nul(\nil),ys:=\nil]$
%  $\gamma = [cs:=\one(\nul(\nil)),ys:=\nil]$,
and
%  $\symb{palindrome}(cs,\nil) \arrz \strue$ with $\gamma = [cs:=
%  \one(\nul(\nil))]$.
$\symb{chk}_\one(\one(xs),\nil) \arrz \strue$ with $\gamma = [xs:=
\nul(\nil)]$.
Consider the first.  As there are no functional variables, $\eta$ is
empty and we need to determine whether
%$\sfalse \in
%\NF^1(\symb{and}(\symb{palindrome}(\one(\nul(\nil)),\nil),
%\symb{chk}_0(\one(\nul(\nil)),\nil,\emptyset))$.
$\strue \in \NF^1(\symb{chk}_\one(\nul(\nil),\nil),\emptyset)$.
This fails, because
$\Conf^0[\xi] = \bot$ for all statements $\xi$.  However, the check
for the second rule, $\strue \in \NF^1(\strue,\emptyset)$, succeeds.
Thus, we mark
%$\Conf^1[\symb{palindrome}(\{\one(\nul( \nil))\},\{\nul(\nil),
%\nil\}) \approx \sfalse] = \top$.
$\Conf^1[\symb{chk}_\one(\{\one(\nul(\nil))\},\{\nul(\nil),\nil\})
\approx \strue] = \top$.
\end{example}

\edef\soundnessthm{\number\value{theorem}}
\begin{theorem}\label{thm:algorithmsoundcomplete}
Let $f : [\asort_1 \times \dots \times \asort_n] \arrtype \bsort \in
\Defineds$ and $s_1 : \asort_1,\dots,s_n : \asort_n, t : \bsort$ be
data terms.  Then $\Conf^I[f(\{s_1\},\dots,\{s_n\}) \approx t] = \top$
if and only if $f(\vec{s}) \arrr{\Rules} t$.
\end{theorem}

\begin{proof}[Proof Sketch]
Define a labeled variation of $\Rules$:
\[
\Ruleslab =\ 
\{ f_{i+1}(\vec{\ell}) \arrz \symb \labl_i(r) \mid f(
  \vec{\ell}) \arrz r \in \Rules \wedge i \in \nats \} \cup
\{ f_{i+1}(\vec{x}) \arrz f_i(\vec{x}) \mid f \in
  \Defineds \wedge i \in \nats \}
\]
Here $\labl_i$ replaces each defined symbol $f$ by a symbol $f_i$.
Then $\Ruleslab$ is infinite, and $f(\vec{s}) \arrr{\Rules} t$ if{f} some $f_i(\vec{s}) \arrr{\Ruleslab} t$.  Furthermore,
$\arr{\Ruleslab}$ is terminating (even if $\arr{\Rules}$ is not!) as is provable using, e.g., the \emph{Computability Path
Ordering}~\cite{bla:jou:rub:08}.  Thus, $\arr{\Ruleslab}$
is a well-founded binary relation on the set of labeled terms, and we
can hence perform induction.

Consider the arguments passed to $\Conf^i$ in the
recursive process: $\NF^i$ is defined using tests of the form
$\Conf^i[f(\NF^i(s_1,\eta),\dots,\NF^i(s_n,\eta))] = \top$, where
each $\eta(x)$ itself has the form $\NF^j(t,\eta')$.  To formally
describe this, let an \emph{$\NF$-substitution} be recursively defined
as a mapping from some (possibly empty) set $V \subseteq \V$ such that
for each $x : \atype \in V$ there are an $\NF$-substitution $\delta$ and a
term $s$ with $\vdash s : \atype$ such that $\eta(x) = \NF^j(s,
\delta)$ for some $j$.
For an $\NF$-substitution $\eta$ on domain $V$, we define
$\overline{\eta}(x) = x$ for $x \notin V$, and $\overline{\eta}(x) =
\labl_j(s)\overline{\zeta}$ for $x \in V$ with $\eta(x) =
\NF^j(s,\zeta)$.  Then the following two claims can be derived by
mutual induction on $q$ ordered with $\arr{\Ruleslab} \cup\:\supterm$
(all $\eta_j$ and $\zeta$ are $\NF$-substitutions):
\begin{itemize}
\item $\Conf^i[f(\NF^{j_1}(s_1,\eta_1),\dots,\NF^{j_n}(s_n,\eta_n))
  \approx t] = \top$ \emph{if and only if} \\
  $q := f_i(\labl_{j_1}(s_1)\overline{\eta_1},\dots,\labl_{j_n}(s_n)
  \overline{\eta_n}) \arrr{\Ruleslab} t$;
\item $t \in \NF^i(u,\zeta)(\NF^{j_1}(s_1,\eta_1),\dots,\NF^{j_n}(s_n,
  \eta_n))$ \emph{if and only if} \\
  $q := \apps{(\labl_i(u)\overline{\zeta})}{\labl_{j_1}(s_1)
  \overline{\eta_1}}{\labl_{j_n}(s_n)\overline{\eta_n}} \arrr{
  \Ruleslab} t$.
\end{itemize}
Since, if we refrain from stopping the process in step $I$, we have
$\Conf^I = \Conf^{I+1} = \Conf^{I+2} = \dots$, the theorem
follows because $f(\vec{s}) \arrr{\Rules} t$ if{f} some $f_i(\vec{s}) \arrr{\Ruleslab} t$.
\end{proof}

It remains to prove that Algorithm~\ref{alg:main} runs sufficiently
fast.

\edef\complexitythm{\number\value{theorem}}
\begin{theorem}\label{thm:algcomplexity}
If $(\F,\Rules)$ has order $k$, then Algorithm~\ref{alg:main} runs in time
$O(\mathrm{exp}_2^k(m \cdot n))$ for some $m$.
\end{theorem}

\begin{proof}
Write $N := |\B|$. As $\Rules$ and $\F$ are fixed, $N$ is linear in
the size of the only input, $s$.
We claim that if
$k,i \in \nats$ are such that $\atype$ has at most \emph{order} $k$, and
the \emph{longest sequence} $\atype_1 \arrtype \dots \arrtype \atype_n
\arrtype \asort$ occurring in $\atype$ has length $n+1 \leq i$, then
$\card(\interpret{\atype}) \leq \mathrm{exp}_2^{k+1}(i^k \cdot N)$.
\vspace{2pt}

(Proof of claim.)
Observe first
that $\P(\B)$ has cardinality $2^N$.
Proceed by induction on the form of $\atype$. Note that we can
write $\atype$ in the form $\atype_1 \arrtype \dots \arrtype \atype_n
\arrtype \asort$ with $n < i$ and each $\atype_j$ having order at
most $k-1$ (as $n = 0$ when given a $0^{\text{th}}$-order type). We have:
{\small
\begin{align*}
\card(\interpret{\atype_1 \arrtype \dots \arrtype \atype_n \arrtype
\asort}) &=
\card((\cdots (\interpret{\asort}^{\interpret{\atype_n}})^{
\interpret{\atype_{n-1}}}\cdots)^{\interpret{\atype_1}}) 
= \card(\interpret{\asort})^{\card(\interpret{\atype_n})\cdots
  \card(\interpret{\atype_1})} \\
&\leq  2^{N \cdot \card(\interpret{\atype_n}) \cdots
  \card(\interpret{\atype_1})} 
\leq  2^{N \cdot \mathrm{exp}_2^k(i^k \cdot N) \cdots
  \mathrm{exp}_2^k(i^k \cdot N)}\hfill
  (\text{by IH})  \\
&= 2^{N \cdot \mathrm{exp}_2^k(i^k \cdot N)^n} 
\leq  2^{\mathrm{exp}_2^k(i^k \cdot N \cdot n + N)}\hfill
  (\text{by induction on}\ k)  \\
&=  \mathrm{exp}_2^{k+1}(n \cdot i^k \cdot N + N)
  \leq \mathrm{exp}_2^{k+1}(i \cdot i^k \cdot N)
  = \exp_2^{k+1}(i^{k+1} \cdot N) \\
& \ \ \ \ (\text{because}\ n \cdot i^k + 1 \leq (n+1) \cdot i^k \leq i \cdot i^k)
\end{align*}
\vspace{-18pt}
}

\noindent (End of proof of claim.)

Since, in a $k^{\text{th}}$-order AFS, all types occurring in type
declarations have order at most $k-1$, there is some $i$ (depending
solely on $\F$) such that all sets $\interpret{\atype}$ in the
algorithm have cardinality $\leq \mathrm{exp}_2^k(i^{k-1} \cdot
N)$.  Writing $a$ for the maximal arity in $\F$, there
are at most $|\Defineds| \cdot \mathrm{exp}_2^k(i^{k-1} \cdot N)^a
\cdot N \leq |\Defineds| \cdot \mathrm{exp}_2^k((i^{k-1} \cdot a +
1) \cdot N)$ distinct statements $f(\vec{A}) \approx t$.

Writing $m := i^{k-1} \cdot a + 1$ and $X := |\Defineds| \cdot
\mathrm{exp}_2^k(m \cdot N)$, we thus find: the algorithm has at most
$I \leq X+2$ steps, and in each step we consider at most $X$
statements $\varphi$ where $\Conf^i[\varphi] = \bot$.  For every
applicable rule, there are at most $(2^N)^a$ different substitutions
$\gamma$% (for every $\ell_j$ which is not a variable, $\ell_j\gamma$
%must be one of the $\leq 2^N$ elements of $A_j$)
, so we have to test
a statement %of the form
$t \in \NF^i(r\gamma,\eta)$ at most $X \cdot (X+2)
\cdot |\Rules| \cdot 2^{aN}$ times.  The exact cost of calculating
$\NF^i(r\gamma,\eta)$ is implementation-specific, but is certainly
bounded by some polynomial $P(X)$ (which depends on the form of $r$).
This leaves the total time cost of the algorithm at $O(X \cdot (X+1)
\cdot 2^{aN} \cdot P(X)) = O(P'(\mathrm{exp}_2^k(m \cdot N)))$ for
some polynomial $P'$ and constant $m$. As $\etime{k}$ is robust under
taking polynomials, the result follows.
\end{proof}

\begin{theorem}\label{thm:characterization}
Let $k \geq 1$. A set $S \subseteq \{0,1\}^+$ is in $\etime{k}$ if{f} there is an AFS of order $k$ that accepts $S$.
\end{theorem}

\begin{proof}
If $S \in \etime{k}$, Theorem~\ref{thm:simulation}
shows that it is accepted by an AFS of order $k$.
Converse\-ly, if there is an AFS of order $k$ that accepts $S$,
Theorem \ref{thm:algcomplexity} shows that we can find whether any basic term reduces to $\strue$ in time
$O(\mathrm{exp}_2^k(m \cdot n))$ for some $m$, and thus $S \in \etime{k}$.
\end{proof}

\begin{remark}
Observe that Theorem~\ref{thm:characterization}
concerns \emph{extensional} rather than \emph{intensional} behaviour of cons-free AFSs: 
a cons-free AFS may take arbitrarily many steps to reduce its input to normal form, even
if it accepts a set that a Turing machine may decide in a bounded number of steps. However,
Algorithm~\ref{alg:main} can often find the possible results of an
AFS faster than evaluating the AFS would take, by avoiding duplicate
calculations.
\end{remark}

\section{Changing the restrictions}\label{sec:changes}

In the presence of non-determinism, minor syntactical changes can make
a large difference in expressivity.  We briefly consider two natural
changes here.

\subsection{Non-left-linearity}

Recall that we imposed three restrictions: the rules in $\Rules$ must
be \emph{constructor rules}, \emph{left-linear} and \emph{cons-free}.
Dramatically, dropping the restriction on left-linearity allows us
to decide every Turing-decidable set using first-order systems.
This is demonstrated by the first-order AFS in
Figure~\ref{fig:nonlinear} which simulates an arbitrary Turing
Machine on input alphabet $I = \{0,1\}$.  Here, a tape $x_0 \dots x_n
\blank\blank \dots$ with the tape head at
position $i$ is represented by a triple $(x_{i-1}\symb{::}\cdots
\symb{::}x_0,~x_i,~x_{i+1}\symb{::}\cdots\symb{::}x_n)$, where the
``list constructor'' $\symb{::}$ is a \emph{defined symbol}, ensured
by a rule which never fires.  To split such a list into a head and
tail, the AFS non-deterministically generates a \emph{new} head and
tail, makes sure they are fully evaluated, and uses a non-left-linear
rule to test whether their combination corresponds to the original
list.
%An additional equality
%check guarantees that the new tail has been fully evaluated to its
%``list'' normal form.

\begin{figure}[htb]
\vspace{-12pt}
\[
\begin{array}{rclrclrcl}
\bot\symb{::}t & \arrz & t &
\symb{rnd} & \arrz & \symb{I} &
\symb{translate}(\symb{0}(xs)) & \arrz & \symb{O::translate}(xs) \\
\symb{rnd} & \arrz & \symb{O} &
\symb{rnd} & \arrz & \symb{B} &
\symb{translate}(\symb{1}(xs)) & \arrz & \symb{I::translate}(xs) \\
& & & 
& & & 
\symb{translate}(\nil) & \arrz & \symb{B::translate}(\nil) \\
\multicolumn{6}{l}{\symb{rndtape}(x) \arrz \nil} &
\symb{translate}(\nil) & \arrz & \nil \\
\multicolumn{6}{l}{\symb{rndtape}(x) \arrz \symb{rnd::rndtape}(x)} &
\symb{equal}(xl,xl) & \arrz & \symb{true} \\
\end{array}
\]
\vspace{-8pt}
\[
\begin{array}{rcl}
\symb{start}(cs) & \arrz & \symb{run}(\symb{startstate},\nil,\symb{B},
  \symb{translate}(cs)) \\
\symb{run}(\unknown{s},xl,\unknown{r},yl) & \arrz &
  \symb{shift}(\unknown{t},xl,\unknown{w},yl,\unknown{d})\ 
  \ \ 
  \llbracket\text{for every transition}\ \transition{\unknown{s}}{
  \unknown{r}}{\unknown{w}}{\unknown{d}}{\unknown{t}}\rrbracket \\
\symb{shift}(s,xl,c,yl,d) & \arrz & \symb{shift}_1(s,xl,c,yl,d,
  \symb{rnd},\symb{rndtape}(\symb{O}),\symb{rndtape}(\symb{I})) \\
\symb{shift}_1(s,xl,c,yl,d,\unknown{b},t,t) & \arrz &
  \symb{shift}_2(s,xl,c,yl,d,\unknown{b},t)\ \ \llbracket
  \text{for every}\ \unknown{b} \in \{\symb{O},\symb{I},\symb{B}\}
  \rrbracket \\
\symb{shift}_2(s,xl,c,yl,\symb{R},z,t) & \arrz &
  \symb{shift}_3(s,c~\symb{::}~xl,z,t,
  \symb{equal}(yl,z~\symb{::}~t)) \\
\symb{shift}_2(s,xl,c,yl,\symb{L},z,t) & \arrz &
  \symb{shift}_3(s,t,z,c~\symb{::}~yl,\symb{equal}(xl,z~\symb{::}~
  t)) \\
\symb{shift}_3(s,xl,c,yl,\symb{true}) & \arrz & \symb{run}(s,xl,c,yl) \\
\end{array}
\]
\caption{A first-order non-left-linear AFS that simulates a Turing machine}
\label{fig:nonlinear}
\end{figure}

\subsection{Product Types}\label{subsec:products}

Unlike AFSs, Jones' minimal language in~\cite{jon:01} employs a 
\emph{pairing constructor}, essentially admitting terms $(s,t) :
\asort \times \bsort$ if $\vdash s : \asort$ and $\vdash t : \bsort$
are data terms or themselves pairs.  This is not in conflict with the
cons-freeness requirement due to type restrictions: it does not allow
construction of an arbitrarily large structure of fixed type.
%In addition, as
%the language in~\cite{jon:01} is deterministic, pairing can always be
%simulated much like we did in Section~\ref{sec:counting}.
%
%In the \emph{non}-deterministic setting with no restriction on evaluation order, however,
In our (non-deterministic) setting, however,
pairing is significantly more powerful.  Following the ideas of
Section~\ref{sec:counting}, one can count up to arbitrarily large
numbers: for an input string $x_n(\dots(x_1(\nil)))$ of length $n$,
\begin{itemize}
\item the counting module $C_0$ represents $i \in \{0,\dots,
  n\}$ by a substring $x_i(\dots(x_1(\nil))) : \bits$;
\item given a $(\lambda n.\mathrm{exp}_2^k(n+1))$-counting module
  $C_k$, %on type $\atype_k$,
  we let $C_{k+1}$ represent a number
  $b$ with bit representation $b_0\dots b_N$ (for $N < \mathrm{exp
  }_2^k(n+1)$) as the pair $(s,t)$---a term!---where
  $s$ reduces to representations of those bits set to $1$, and
  $t$ to representations of bits set to $0$.
\end{itemize}
Then for instance a number in $\{0,\dots,2^{2^{n+1}}-1\}$ is
represented by a pair $(s,t) : (\bits \times \bits)
\times (\bits \times \bits)$, where $s$ and $t$
themselves are \emph{not} pairs; rather, they are both terms reducing
to a variety of different pairs.  A membership test would take the
form
\[
\begin{array}{ccc}
\multicolumn{3}{c}{\symb{elem}_2(k,(s,t)) \arrz
  \symb{elemtest}(\symb{equal}_1(k,s),\symb{equal}_1(k,t))} \\
\symb{elemtest}(\symb{true},x) \arrz \symb{true} &
\phantom{ABCD} &
\symb{elemtest}(x,\symb{true}) \arrz \symb{false} \\
\end{array}
\]
with the rule for $\symb{equal}_1$ having the form $\symb{equal}_1(
(s_1,t_1),(s_2,t_2)) \arrz r$.  That is, the rule \emph{forces a
partial evaluation}.  This is possible because a ``false constructor''
(i.e., a syntactic structure that rules can match) is allowed to occur above non-data terms.

\section{Future work}\label{sec:conclusion}

\CKchange{In this paper, we have considered the expressive power of cons-free
term rewriting, and seen that restricting data order results in
characterizations of different classes.
A natural direction for future work is to consider further restrictions,
either on rule formation, reduction strategy, or both.}
%We believe it to be of interest to recover
%extant characterizations of similar hierarchies by suitable syntactic
%restrictions.
%For example, Jones \cite{jon:01} obtains the hierarchy
%$\textrm{P} \subseteq \textrm{EXPTIME} \subseteq \exptime{2}
%\subsetneq \cdots$ for his call-by-value language, and we believe we
%would obtain the exact same hierarchy by restricting to innermost
%evaluation---the main novelty value being that non-determinism---as
%usual \cite{DBLP:journals/jacm/Cook71,DBLP:conf/amast/Bonfante06,car:sim:14}---would
%not allow for non-deterministic classes, as already known for
%first-order systems
Following Jones~\cite{jon:01}, we suspect that
\CKchange{restricting to innermost evaluation will give}
%we would obtain
the hierarchy $\textrm{P} \subseteq \textrm{EXPTIME} \subseteq
\exptime{2} \subsetneq \cdots$.
%by restricting to innermost evaluation.
\CKchange{Furthermore, we}
%But also, since minor syntactic changes can make a large
%difference (see Section~\ref{sec:changes}) we
conjecture
that a combination of higher-order rewriting and restrictions on
rule formation, possibly together with additions such as product
types, may yield characterizations of a \CKchange{wide range} of
classes, including \CKchange{non-deterministic classes like
$\textrm{NP}$ or very small classes like $\textrm{LOGTIME}$}.
%However, as do the authors of \cite{car:sim:14},
%we also conjecture that restrictions on rule formation
%(e.g., mimicking restricted notions of recursions) combined with
%higher-order rewriting may yield characterizations of
%non-deterministic classes like $\textrm{NP}$ and $\textrm{NE}$.

%Finally, while cons-freeness does not naturally lend itself to producing output, it is interesting to investigate
%characterizations of sets of computable functions, e.g. the polytime-computable functions on naturals, rather than
%decidable sets.

%\subparagraph*{Acknowledgements}
%
%The authors thank \dots

\bibliography{finalbib}

\newpage

\appendix

% appendix counters
\edef\savedcounter{\number0}
\newcommand{\startappendixcounters}{
  \setcounter{theorem}{\savedcounter}
  \renewcommand{\thetheorem}{\Alph{section}\arabic{theorem}}
}
\newcommand{\oldcounter}[1]{
  \edef\savedcounter{\number\value{theorem}}
  \setcounter{theorem}{#1}
  \renewcommand{\thetheorem}{\arabic{theorem}}
}
\startappendixcounters

\section{Proofs omitted from the main text}

In Section \ref{sec:prelims} we claim that for constructor rewriting systems
$\Rules$ the following holds:

\begin{center}
If $t\:_\beta\!\leftarrow s \arrr{\Rules} q$ with $q$ a normal
form, then $t \arrr{\Rules} q$ as well.
\end{center}

This is used as justification to not consider rules with a
$\beta$-redex $\app{(\abs{x}{s})}{t}$ in the right-hand side.  We
will obtain this result as an easy consequence of the labeled system
employed for the proofs in Section~\ref{sec:algorithm}.  Thus, in this
appendix we will allow such rules until the claim is proven in
Lemma~\ref{lem:betafirst}.

\subsection{Proofs of Section~\ref{sec:limitations}}

To facilitate proving the properties on $\B$-safety, we first extend
the definition to be parametrized over a set of proper constructor
terms satisfying certain rules.  In the following, we assume that
$\B$ is a set of data terms which is closed under $\supterm$ and
contains all data terms occurring in the right-hand side of a rule in
$\Rules$.

\begin{definition}[$\B^X$-safety]
Let $X$ be a set of proper constructor terms on disjoint variables,
which does not contain any variable occurring bound in $s$; then:
\begin{enumerate}[label=(\Alph*)]
\item\label{bsafe:X} any subterm $s \subtermeq t \in X$ is
  $\B^X$-safe;
\item\label{bsafe:base} any term in $\B$ is $\B^X$-safe;
\item\label{bsafe:var} any variable is $\B^X$-safe;
\item\label{bsafe:fun} if $f \in \Defineds$ and $s_1,\dots,s_n$ are
  $\B^X$-safe, then $f(s_1,\dots,s_n)$ is $\B^X$-safe (if well-typed);
\item\label{bsafe:app} if $s$ and $t$ are both $\B^X$-safe,
  then $\app{s}{t}$ is $\B^X$-safe (if well-typed);
\item\label{bsafe:abs} if $x \in \V$ and $s$ is $\B^X$-safe, then 
  $\abs{x}{s}$ is $\B^X$-safe.
\end{enumerate}
\end{definition}

It is easy to see that a term is $\B$-safe if{f} it is $\B^\emptyset$-safe. 
Note also that if we $\alpha$-rename all rules to make sure the same variables do not
occur both bound and free, then the right-hand side $r$ of a cons-free
rule $f(\vec{\ell}) \arrz r$ is $\B^{\{\vec{\ell}\}}$-safe.

We have the following properties:

\begin{lemma}\label{lem:safetyprop}
For all $\B^X$-safe terms $s$:
\begin{enumerate}
\item\label{bsafe:prop:subterm} all subterms $t$ of $s$ are
  $\B^X$-safe;
\item\label{bsafe:prop:substitute} if $\gamma$ is a substitution
  such that $t\gamma$ is $\B$-safe for all $t \in X \cup \FV(s)$,
  then $s\gamma$ is $\B$-safe.
\end{enumerate}
\end{lemma}

\begin{proof}
All three properties follow by a simple induction on the form of
$s$.  Note that for the second property, all variables in $s$ are
renamed to fresh ones beforehand, which therefore do not occur
anywhere in $X$ or in the domain or range of $\gamma$.
\begin{itemize}
\item \emph{property (\ref{bsafe:prop:subterm}):} For case
  \ref{bsafe:base} we note that $\B$ is closed under subterms; the
  other cases are obvious.
\item \emph{property (\ref{bsafe:prop:substitute}):} 
  Case \ref{bsafe:X} holds by property (\ref{bsafe:prop:subterm}):
  $s \subtermeq t \in X$ implies $s\gamma \subtermeq t\gamma$, which
  is $\B$-safe by assumption.  Case \ref{bsafe:base} holds because all
  elements of $\B$ are closed, so $s\gamma = s \in \B$.  Case
  \ref{bsafe:var} follows by assumption, and cases
  \ref{bsafe:fun}--\ref{bsafe:app} by the induction hypothesis.
  \qedhere
\end{itemize}
\end{proof}

We recall Lemma~\ref{lem:safetypreserve}, the primary property of
interest for $\B$-safety:

\oldcounter{\safetypreservelem}
\begin{lemma}
If $s$ is $\B$-safe and $s \arr{\Rules} t$, then $t$ is $\B$-safe.
\end{lemma}
\startappendixcounters

\begin{proof}
By induction on the form of $s$.
First suppose the reduction does not take place at the root.  Since
$s$ reduces, it cannot be a variable or data term, so it has one of
three forms:
\begin{itemize}
\item $s = f(s_1,\dots,s_n)$ with $f \in \Defineds$ and all $s_i$
  are $\B$-safe.  Then the
  reduction takes place in some $s_i$, so $t = f(s_1,\dots,s_i',
  \dots,s_n)$ with $s_i \arr{\Rules} s_i'$, so also $s_i'$ is
  $\B$-safe by induction.  This, and $\B$-safety of all other
  $s_j$, gives $\B$-safety of $t$.
\item $s = \app{u}{v}$.  Then either $t = \app{u'}{v}$ with $u
  \arr{\Rules} u'$ (and therefore $u'$ is $\B$-safe) or $t =
  \app{u}{v'}$ with $v \arr{\Rules} v'$ (and therefore $v'$ is
  $\B$-safe).  Either way, $t$ is the application of two
  $\B$-safe terms and therefore $\B$-safe.
\item $s = \abs{x}{u}$.  In this case, the reduction must take place
  in the $\B$-safe term $u$, so $t = \abs{x}{u'}$ and $u'$ is
  $\B$-safe as well by induction; $\B$-safety of $t$ follows.
\end{itemize}
This leaves the base case, a reduction at the root.  Here, there are
two possibilities:
\begin{itemize}
\item $s = \app{(\abs{x}{u})}{v}$ and $t = u[x:=v]$.  By $\B$-safety
  of $s$, also $u$ and $v$ are $\B$-safe, so by
  Lemma~\ref{lem:safetyprop}(\ref{bsafe:prop:substitute}) the
  result $t$ is $\B$-safe as well.
\item $s = \ell\gamma$ and $t = r\gamma$ for some rule $\ell \arrz r
  \in \Rules$ and substitution $\gamma$ which maps $\ell$ to a
  $\B$-safe term.  Writing $\ell = f(\vec{\ell})$, we can assume that
  $r$ is $\alpha$-renamed to be $\B^{\{\vec{\ell}\}}$-safe, so by
  Lemma~\ref{lem:safetyprop}(\ref{bsafe:prop:substitute}) we obtain
  $\B$-safety of $t$.
  \qedhere
\end{itemize}
\end{proof}

\subsection{Proofs of Section~\ref{sec:counting}}

We move on to the results of Section~\ref{sec:counting}.

\oldcounter{\prodmodule}
\begin{lemma}
If there exist a $P$-counting module $C_\pi$ and a $Q$-counting module
$C_\rho$, both of order at most $k$, then there is a $\lambda n.P(n)
\cdot Q(n)$-counting module $C_{\pi\cdot\rho}$ of order at most $k$.
\end{lemma}
\startappendixcounters

\begin{proof}
Let $C_\pi ::= ([\atype_1 \times \dots \times \atype_a],\Sigma^\pi,
R^\pi,A^\pi,\numinterpret{\cdot}^\pi)$ and $C_\rho ::= ([\btype_1
\times \dots \times \btype_b],\Sigma^\rho,R^\rho,A^\rho,
\numinterpret{\cdot}^\rho)$.  We can safely assume that any symbol
$f$ which occurs in both $\Sigma^\pi$ and $\Sigma^\rho$ has the same
type declaration in both, and is defined by the same rules in $R^\pi$
and $R^\rho$---if this is not the case, we simply use a renaming.
Thus, we are given two counting modules that have no
\emph{conflicts}: combining the signatures and rules does not affect
the reduction and interpretation properties.

Let $C_{\pi \cdot \rho} = ([\atype_1 \times \dots \times
\atype_a \times \btype_1 \times \dots \times \btype_b],\Sigma^\pi
\cup \Sigma^\rho \cup \Sigma,R^\pi \cup R^\rho \cup R,A^{\pi \cdot
\rho},\numinterpret{\cdot}^{\pi \cdot \rho})$, where:
\begin{itemize}
\item $A^{\pi \cdot \rho} = \{ (u_1,\dots,u_a,v_1,\dots,v_b) \mid
  (u_1,\dots,u_a) \in A^\pi \wedge (v_1,\dots,v_b) \in A^\rho \}$,
\item $\numinterpret{(u_1,\dots,u_a,v_1,\dots,v_b)}^{\pi \cdot \rho}_{
  cs} = \numinterpret{(u_1,\dots,u_a)}^\pi_{cs} \cdot Q(|cs|) +
  \numinterpret{(v_1,\dots,v_b)}^\rho_{cs}$,
\item $\Sigma$ consists of the defined symbols introduced in $R$,
  which we construct below.
\end{itemize}

Intuitively, fixing $cs$ and writing $N := P(|cs|)$ and $M :=
Q(|cs|)$, a number $i$ in $\{0,\dots,N \cdot M - 1\}$ can be seen as
a unique pair $(n,m)$ with $0 \leq n < N$ and $0 \leq m < M$, such
that $i = n \cdot m$.  Here, $n$ is represented by a tuple
$(u_1,\dots,u_a)$ in the counting module $C_\pi$, and $m$ by a tuple
$(v_1,\dots,v_b)$ in $C_\rho$.

For the $\symb{seed}$ function, we observe that $N \cdot M - 1 = (N -
1) \cdot M + (M - 1)$, which corresponds to the pair $(N-1,M-1)$,
which in turn translates to the tuple
$(\symb{seed}_\pi^1[cs],\dots,\symb{seed}_\pi^a[cs],\linebreak
\symb{seed}_\rho^1[cs],\dots,\symb{seed}_\rho^b[cs])$.  This tuple
is generated by the following rules:
\[
\begin{array}{rcl}
\symb{seed}_{\pi \cdot \rho}^i(cs,\vec{z}) & \arrz &
  \symb{seed}_\pi^i(cs,\vec{z})\ \ \text{for}\ 1 \leq i \leq a \\
\symb{seed}_{\pi \cdot \rho}^i(cs,\vec{z}) & \arrz & \symb{seed}_{
  \pi \cdot \rho}^{i-a}(cs,\vec{z})\ \ \text{for}\ a+1 \leq i \leq
  a+b \\
\end{array}
\]
Note the extra parameters $\vec{z}$: this we do because some
$\atype_i$ may be a functional type, and all functions have a sort as
output type (as observed in the definition of counting modules).

The $\symb{zero}$ function requires both components to be $0$:
\[
\begin{array}{rcl}
\symb{zero}_{\pi \cdot \rho}(cs,u_1,\dots,u_a,v_1,\dots,v_b) & \arrz &
  \symb{and}(\symb{zero}_\pi(cs,u_1,\dots,u_a),
  \symb{zero}_\rho(cs,v_1,\dots,v_b)) \\
\end{array}
\]
\vspace{-18pt}
\[
\begin{array}{rclcrcl}
\symb{and}(\strue,x) & \arrz & x & &
\symb{and}(\sfalse,y) & \arrz & \sfalse \\
\end{array}
\]
For inverses, note that $N \cdot M - (n \cdot M + m) - 1 =
(N - m) \cdot M - m - 1 = (N - m - 1) \cdot M + N - m - 1$, giving the
pair $(N-n-1,M-m-1)$, or $(\symb{inv}(n),\symb{inv}(m))$:
\[
\begin{array}{rcl}
\symb{inv}_{\pi \cdot \rho}^i(cs,u_1,\dots,u_a,v_1,\dots,v_b,\vec{z})
  & \arrz & \symb{inv}_\pi^i(cs,u_1,\dots,u_a,\vec{z})
  \ \ \text{for}\ 1 \leq i \leq a \\
\symb{inv}_{\pi \cdot \rho}^i(cs,u_1,\dots,u_a,v_1,\dots,v_b,\vec{z})
  & \arrz & \symb{inv}_\rho^{i-a}(cs,v_1,\dots,v_b,\vec{z})\ \ 
  \text{for}\ a+1 \leq i \leq a+b \\
\end{array}
\]
For the predecessor, $(i,j)-1$ results in $(i,j-1)$ if $j > 0$,
otherwise in $(i-1,M-1))$:
\[
\begin{array}{rcl}
\symb{pred}^i_{\pi \cdot \rho}(cs,u_1,\dots,u_a,v_1,\dots,v_b,\vec{z}) & \arrz &
  \symb{ptest}_{\pi \cdot \rho}^i(\symb{zero}_\rho(cs,v_1,\dots,v_b),
  cs,u_1,\dots,u_a, \\
  & & \ \ \ \ \ \ \ \ \ \ \ \ \ 
  v_1,\dots,v_b,\vec{z})\ \ \text{for}\ 1 \leq i \leq a+b \\
\symb{ptest}^i_{\pi \cdot \rho}(\symb{false},cs,\vec{u},\vec{v},
  \vec{z}) & \arrz & \app{u_i}{\vec{z}}\ \ \text{for}\ 1 \leq i \leq a \\
\symb{ptest}^i_{\pi \cdot \rho}(\symb{false},cs,\vec{u},\vec{v},
  \vec{z}) & \arrz & \symb{pred}_\rho^{i-a}(cs,v_1,\dots,v_b,\vec{z})
  \ \ \text{for}\ a+1 \leq i \leq a+b \\
\symb{ptest}^i_{\pi \cdot \rho}(\symb{true},cs,\vec{u},\vec{v},
  \vec{z}) & \arrz & \symb{pred}_\pi^i(cs,u_1,\dots,u_a,\vec{z})\ \ 
  \text{for}\ 1 \leq i \leq a \\
\symb{ptest}^i_{\pi \cdot \rho}(\symb{true},cs,\vec{u},\vec{v},
  \vec{z}) & \arrz & \symb{seed}_\rho^{a-i}(cs,v_1,\dots,v_b,\vec{z})\ \ 
  \text{for}\ a+1 \leq i \leq a+b \\
\end{array}
\]
Note the use of $\app{v_i}{\vec{z}}$: this rule can be read as
$\symb{ptest}^i_{\pi \cdot \rho}[\sfalse,cs,\vec{u},\vec{v}]
\arr{\Rules} u_i$ if $1 \leq i \leq a$ (modulo $\alpha$-equivalence).

For the successor, $(i,j)+1$ results in $(i,j+1)$ if $j < M-1$,
and in $(i+1,0)$ otherwise.  The former holds exactly if
$\symb{inv}(j)$ is non-zero, and $0$ is exactly
$\symb{inv}(\symb{seed}(cs))$.
\[
\begin{array}{rcl}
\symb{suc}^i_{\pi \cdot \rho}(cs,u_1,\dots,u_a,v_1,\dots,v_b,\vec{z}) & \arrz &
  \symb{suctest}^i_{\pi \cdot \rho}(\symb{zero}_\rho(cs,
  \symb{inv}_\rho^1[cs,\vec{v}],
  \dots,\symb{inv}_\rho^b[cs,\vec{v}]), \\
  & & \hfill
  u_1,\dots,u_a,v_1,\dots,v_b,\vec{z})
  \ \ \text{for}\ 1 \leq i \leq a+b \\
\symb{suctest}^i_{\pi \cdot \rho}(\symb{false},cs,\vec{u},\vec{v},
  \vec{z}) & \arrz & \app{u_i}{\vec{z}}
  \ \ \text{for}\ 1 \leq i \leq a \\
\symb{suctest}^i_{\pi \cdot \rho}(\symb{false},cs,\vec{u},\vec{v},
  \vec{z}) & \arrz & \symb{suc}_\rho^{i-a}(cs,v_1,\dots,v_b,\vec{z})\ \ 
  \text{for}\ a+1 \leq i \leq a+b \\
\symb{suctest}^i_{\pi \cdot \rho}(\symb{true},cs,\vec{u},\vec{v},
  \vec{z}) & \arrz & \symb{suc}_\pi^i(cs,u_1,\dots,u_a,
  \vec{z})\ \ \text{for}\ 1 \leq i \leq a \\
\symb{suctest}^i_{\pi \cdot \rho}(\symb{true},cs,\vec{u},\vec{v},
  \vec{z}) & \arrz & \symb{nul}_\rho^{a-i}(cs,\vec{z})\ \ 
  \text{for}\ a+1 \leq i \leq a+b \\
\symb{nul}^i_\rho(cs,\vec{z}) & \arrz &
  \symb{inv}_\rho^i(cs,\symb{seed}^1_\rho[l],\dots,
  \symb{seed}^b_\rho[l],\vec{z})\ \text{for}\ 1 \leq i \leq b \\
\end{array}
\]
\end{proof}

\oldcounter{\expmodule}
\begin{lemma}
If there is a $P$-counting module $C_{\pi}$ of order $k$, then there
is a $\lambda n.2^{P(n)}$-counting module $C_{p[\pi]}$ of order $k+1$.
\end{lemma}
\startappendixcounters

\begin{proof}
Assume given a $P$-counting module $C_\pi = ([\atype_1 \times
\dots \times \atype_a],\Sigma,R,A,\numinterpret{\cdot}^\pi)$.
We define the $2^P$-counting module $C_{\symb{p}[\pi]}$ as the tuple
$([\atype_1 \arrtype \dots \arrtype \atype_a \arrtype \bool],
\Sigma^{\symb{p}[\pi]},R^{\symb{p}[\pi]},B,
\numinterpret{\cdot}^{\symb{p}[\pi])}$, where:
\begin{itemize}
\item $B_{cs}$ is the set of all terms $q \in \Terms(\Sigma^{\symb{p}[
  \pi]} \cup \Constructors,\emptyset)$ of type $\atype_1 \arrtype
  \dots \arrtype \atype_a \arrtype \bool$ such that:
  \begin{itemize}
  \item for all $(s_1,\dots,s_a) \in A_{cs}$: $\apps{q}{s_1}{s_a}$
    reduces to either $\strue$ or $\sfalse$, but not to both;
  \item for all $(s_1,\dots,s_a),(t_1,\dots,t_a) \in A_{cs}$: if
    $\numinterpret{(\vec{s})}^\pi_{cs} = \numinterpret{(\vec{t})
    }^\pi_{cs}$, then $\apps{q}{s_1}{s_a}$ and $\apps{q}{t_1}{t_a}$
    reduce to the same boolean value.
  \end{itemize}
\item Writing $N := P(|cs|)-1$, let
  $\numinterpret{q}^{\symb{p}[\pi]}_{cs} = \sum_{i=0}^N \{ 2^{N-i
  } \mid \apps{q}{s_1}{s_a} \arrr{R} \strue$ for some
  $(s_1,\dots,s_a)$ with $\numinterpret{(s_1,\dots,s_a)}^\pi_{cs} =
  i \}$; that is, $q$ represents the number given by the bitvector
  $b_0\dots b_N$ (with $b_N$ the least significant digit) where $b_i
  = 1$ if and only if $q \cdot \numrep{i} \arrr{R^{\symb{p}[\pi]}}
  \strue$ for some representation $\numrep{i}$ of $i$ in the counting
  module $C_\pi$ (note that, by the requirement on $B_{cs}$, this
  therefore holds for \emph{any} representation of $i$).
\item $\Sigma^{\symb{p}[\pi]} = \Sigma \cup \Sigma'$ and
  $R^{\symb{p}[\pi]} = R \cup R'$, where $\Sigma'$ consists of the
  defined symbols introduced in $R'$, which we construct below.
\end{itemize}

To start, $\symb{seed}[cs]$ should return a bitvector that is $1$ at
all bits, so having $\symb{seed}[cs] \arrr{R'} \lambda \vec{k}.\strue$
would suffice.  By definition of the $f[\vec{s}]$ construction, that
is:
\[
\begin{array}{rcl}
\symb{seed}_{\symb{p}[\pi]}(cs,k_1,\dots,k_a) & \arrz & \strue \\
\end{array}
\]
The inverse of a bitvector is obtained by flipping all the bits, as
we saw n Lemma~\ref{lem:mainmodule}.  Thus:
\[
\begin{array}{rclcrcl}
\symb{inv}_{\symb{p}[\pi]}(cs,F,k_1,\dots,k_a) & \arrz &
  \symb{not}(\apps{F}{k_1}{k_a}) & & 
\symb{not}(\strue) & \arrz & \sfalse \\
& & & & \symb{not}(\sfalse) & \arrz & \strue \\
\end{array}
\]
For the $\symb{zero}$ function, we simply test whether all bits are
set to $0$:
\[
\begin{array}{rcl}
\symb{zero}_{\symb{p}[\pi]}(cs,F) & \arrz & \symb{zero'}_{\symb{p}[
  \pi]}(cs,\symb{seed}^1_\pi[cs],\dots,\symb{seed}^a_\pi[cs],F) \\
\symb{zero'}_{\symb{p}[\pi]}(cs,k_1,\dots,k_a,F) & \arrz &
  \symb{ztest}_{\symb{p}[\pi]}(\apps{F}{k_1}{k_a},\symb{zero}_\pi(cs,
  k_1,\dots,k_a),cs\\ & & \phantom{ztest_{p[\pi]}}\ \ \ 
  k_1,\dots,k_a,F) \\
\symb{ztest}_{\symb{p}[\pi]}(\strue,z,cs,\vec{k},F) & \arrz &
  \sfalse \\
\symb{ztest}_{\symb{p}[\pi]}(\sfalse,\strue,cs,\vec{k},F) & \arrz &
  \strue \\
\symb{ztest}_{\symb{p}[\pi]}(\sfalse,\sfalse,cs,\vec{k},F) & \arrz &
  \symb{zero'}_{\symb{p}[\pi]}(cs,\symb{pred}^1_\pi[cs,\vec{k}],
  \dots,\symb{pred}^a_\pi[cs,\vec{k}], F)
\end{array}
\]
For the predecessor function, we observe as before that $x_0 \dots x_i
1 0 \dots 0$ has $x_1 \dots x_i 0 1 \dots 1$ as a predecessor; that is,
we must flip all the bits until we encounter a $1$, flip that one
too, and leave the function unmodified for the rest.  To this end, we
first define what it means to flip a bit: we want $\symb{flip}[F,
\vec{k}]$ to be the function that maps $\vec{z}$ to $\app{F}{\vec{z}}$
if $\numinterpret{\vec{k}}^\pi \neq \numinterpret{\vec{z}}^\pi$ and
to $\symb{not}(\app{F}{\vec{z}})$ otherwise.  For this, of course, we
will need to define an equality check as well.
\[
\begin{array}{rcl}
\symb{flip}_{\symb{p}[\pi]}(cs,F,k_1,\dots,k_a,z_1,\dots,z_a) & \arrz
  & \symb{flipcheck}_{\symb{p}[\pi]}(F,\vec{z},\symb{equal}_\pi(cs,
  \vec{k},\vec{z})) \\
\symb{flipcheck}_{\symb{p}[\pi]}(F,z_1,\dots,z_a,\sfalse) & \arrz &
  \apps{F}{z_1}{z_a} \\
\symb{flipcheck}_{\symb{p}[\pi]}(F,z_1,\dots,z_a,\strue) & \arrz &
  \symb{not}(\apps{F}{z_1}{z_a}) \\
\symb{equal}_\pi(cs,k_1,\dots,k_a,z_1,\dots,z_a) & \arrz &
  \symb{eqtest}_\pi(\symb{zero}_\pi(cs,\vec{k}),\symb{zero}_\pi(cs,\vec{z}),
  cs,\vec{k},\vec{z}) \\
\symb{eqtest}_\pi(\strue,b,cs,\vec{k},\vec{z}) & \arrz & b \\
\symb{eqtest}_\pi(\sfalse,\strue,cs,\vec{k},\vec{z}) & \arrz & \sfalse \\
\symb{eqtest}_\pi(\sfalse,\sfalse,cs,\vec{k},\vec{z}) & \arrz &
  \symb{equal}_\pi(cs,\symb{pred}^1_\pi[cs,\vec{k}],\dots,
  \symb{pred}^a_\pi[cs,\vec{k}],\\
& & \phantom{\symb{equal}_\pi(cs,}~\symb{pred}^1_\pi[cs,\vec{z}],\dots,
  \symb{pred}^a_\pi[cs,\vec{z}]) \\
\end{array}
\]
This, we use to define our predecessor function.
\[
\begin{array}{rcl}
\symb{pred}_{\symb{p}[\pi]}(cs,F,\vec{z}) & \arrz & \symb{pred'}_{
  \symb{p}[\pi]}(cs,\symb{seed}^1_\pi[cs],\dots,\symb{seed}^a_\pi[cs],F,
  \vec{z}) \\
\symb{pred'}_{\symb{p}[\pi]}(cs,k_1,\dots,k_a,F,\vec{z}) & \arrz &
  \symb{predtest}_{\symb{p}[\pi]}(\apps{F}{k_1}{k_a},
  \symb{zero}_\pi(cs,\vec{k}),cs,\vec{k},\\
& & \phantom{\symb{predtest}_{\symb{p}[\pi]}(}
  \symb{flip}_{\symb{p}[\pi]}[cs,F,\vec{k}],\vec{z}) \\
\symb{predtest}_{\symb{p}[\pi]}(\strue,b,cs,\vec{k},F,\vec{z}) & \arrz &
  \app{F}{\vec{z}} \\
\symb{predtest}_{\symb{p}[\pi]}(\sfalse,\strue,cs,\vec{k},F,\vec{z}) &
  \arrz & \symb{not}(\app{F}{\vec{z}}) \\
\symb{predtest}_{\symb{p}[\pi]}(\sfalse,\sfalse,cs,\vec{k},F,\vec{z}) &
  \arrz & \symb{pred'}_{\symb{p}[\pi]}(cs,\symb{pred}^1_\pi[cs,\vec{k}],
  \dots,\symb{pred}^a_\pi[cs,\vec{k}],F,\vec{z}) \\
\end{array}
\]
Note the way $\symb{not}$ is used in the second-last rule: this is
the case where we continue flipping bits until $b_0$ is reached, and
$b_0$ itself is $0$; that is, the number represented by $F$ is $0$.
As the $\symb{pred}$-function iteratively updates the functional
argument, this argument will return $\strue$ at all positions by the
time this last bit is reached.  That is why $\symb{not}$ is applied.

Finally, the successor function is obtained by combining $\symb{inv}$
and $\symb{pred}$ as in Lemma~\ref{lem:mainmodule}.
\[
\symb{suc}_{\symb{p}[\pi]}(cs,F,\vec{z}) \arrz
\symb{inv}_{\symb{p}[\pi]}(cs,\symb{pred}_{\symb{p}[\pi]}[cs,
\symb{inv}_{\symb{p}[\pi]}[cs,F]],\vec{z})
\qedhere
\]
%can be done in much the same way as
%the predecessor, since the successor of a bitvector $x_0 \dots x_i 0
%1 \dots 1$ is $x_0 \dots x_i 1 0 \dots 0$.
%\[
%\begin{array}{rcl}
%\symb{suc}_{\symb{p}[x]}(l,F,\vec{z}) & \arrz &
%  \symb{suc'}_{\symb{p}[x]}(l,\symb{seed}^1_x[l],\dots,
%  \symb{seed}^a_x[l],F,\vec{z}) \\
%\symb{suc'}_{\symb{p}[x]}(l,k_1,\dots,k_a,F,\vec{z}) & \arrz &
%  \symb{suctest}_{\symb{p}[x]}(\apps{F}{k_1}{k_a},
%  \symb{zero}_x(l,\vec{k}),l,\vec{k},\\
%& & \phantom{\symb{suctest}_{\symb{p}[x]}(}
%  \symb{flip}[l,F,\vec{k}],\vec{z}) \\
%\symb{suctest}_{\symb{p}[x]}(\sfalse,b,l,\vec{k},F,\vec{z}) & \arrz &
%  \app{F}{\vec{z}} \\
%\symb{suctest}_{\symb{p}[x]}(\strue,\strue,l,\vec{k},F,\vec{z}) &
%  \arrz & \symb{not}(\app{F}{\vec{z}}) \\
%\symb{suctest}_{\symb{p}[x]}(\strue,\sfalse,l,\vec{k},F,\vec{z}) &
%  \arrz & \symb{suc'}_{\symb{p}[x]}(l,\symb{pred}^1_x[l,\vec{k}],
%  \dots,\symb{pred}^a_x[l,\vec{k}],F,\vec{z}) \\
%\end{array}
%\]
\end{proof}

\subsection{Proofs of Section~\ref{sec:algorithm}}

In Section~\ref{sec:algorithm}, we must prove correctness of
the algorithm (Theorem~\ref{thm:algorithmsoundcomplete}).  This proof
takes several large steps.  To start, we introduce a terminating
counterpart to $\Rules$.

\begin{definition}[Labeled System]
Let
$
\Flab := \F \cup \{ f_i : \alpha \mid f : \alpha \in \Defineds \wedge
  i \in \nats \}
$.
For $s \in \Terms(\F,\V)$ and $i \in \nats$, let $\labl_i(s)$ be $s$
with all instances of a defined symbol $f$ replaced by $f_i$.  For
$t \in \Terms(\Flab,\V)$, let $|t|$ be $t$ with all symbols $f_i$
replaced by $f$.  Then, let
\[
\begin{array}{rl}
\Ruleslab =\ & \{ f(x_1,\dots,x_n) \arrz f_i(x_1,\dots,x_n) \mid
  f : [\atype_1 \times \dots \times \atype_n] \arrtype \asort \in
  \Defineds \wedge i \in \nats \}\ \cup \\
& \{ f_{i+1}(x_1,\dots,x_n) \arrz f_i(x_1,\dots,x_n) \mid f :
  [\atype_1 \times \dots \times \atype_n] \arrtype \asort \in
  \Defineds \wedge i \in \nats \}\ \cup \\
& \{ f_{i+1}(\ell_1,\dots,\ell_n) \arrz \labl_i(r) \mid
  f(\vec{\ell}) \arrz r \in \Rules \wedge i \in \nats \}
\end{array}
\]
\end{definition}

Note that constructor terms are unaffected by $\labl_i$ and $|\cdot|$.
While the AFS $(\Flab,\Ruleslab)$ is obviously non-deterministic and
infinite in both its signature and rules, these issues do not block us
from using it as a reasoning tool.  Importantly, this AFS defines the
same decision function as $(\F,\Rules)$:

\begin{lemma}\label{lem:labeledequiv}
For all $f : [\atype_1 \times \dots \times \atype_n] \arrtype \asort
\in \Defineds$ and data terms $s_1,\dots,s_n,t$:
\[
f(s_1,\dots,s_n) \arrr{\Rules} t\ 
\text{if and only if}\ 
f(s_1,\dots,s_n) \arrr{\Ruleslab} t
\]
\end{lemma}

\begin{proof}
For the \emph{if} direction, note that:
\begin{itemize}
\item if %$\gamma$ and $\gamma^{|\cdot|}$ are substitutions with
  $\gamma^{|\cdot|}(x) = |\gamma(x)|$ for all $x$, then 
  $|u\gamma| = |u|\gamma^{|\cdot|}$ for all $u \in \Terms(\Flab,\V)$;
\item therefore, if $u = \ell\gamma$ and $v = r\gamma$ for
  $\ell \arrz r \in \Ruleslab$, then either $|u| = |\ell|\gamma^{|
  \cdot|} = |r|\gamma^{|\cdot|} = |v|$ (for the first two groups of
  rules, where $|\ell| = |r|$), or $|u| = |\ell|\gamma^{|\cdot|}
  \arr{\Rules} |r|\gamma^{|\cdot|} = |v|$;
\item therefore, if $u \arr{\Ruleslab} v$ either
  $|u| \arr{\Rules} |v|$ or $|u| = |v|$.
\end{itemize}
The first and third observations follow by a straightforward induction
on $u$, the second by definition of $\Ruleslab$ and $|\cdot|$.  The
\emph{if} statement now follows straightforwardly by induction on the
length of the derivation $f(\vec{s}) \arrr{\Ruleslab} t$.

For the \emph{only if} direction, proceed as follows. For any substitution $\gamma$ and $i \in \nats$,
let $\gamma_i$ be the substitution mapping each $x$ to
$\labl_i(\gamma(x))$.  We observe:
\begin{itemize}
\item for all $s$, $\gamma$, $i$:
  $\labl_i(s\gamma) = \labl_i(s)\gamma_i$ (by structural induction on
  $s$);
\item for all $s,i$: $\labl_{i+1}(s) \arrr{\Ruleslab}
  \labl_i(s)$ (by structural induction on $s$);
\item therefore, if $u = f(\vec{\ell})\gamma$ and $v =
  r\gamma$ with $f(\vec{\ell}) \arrz r \in \Rules$, then
  $\labl_{i+1}(u) = \linebreak
  \labl_{i+1}(\ell)\gamma_{i+1}
  \arrr{\Ruleslab} f_{i+1}(\vec{\ell})
  \gamma_i \arr{\Ruleslab} \labl_i(r)\gamma_i = \labl_i(v)$;
\item therefore, if $u \arr{\Rules} v$, then $\labl_{i+1}(u)
  \arrr{\Rules} \labl_i(v)$ (by structural induction on $u$);
\item thus, if $f(\vec{s}) \arrr{\Rules} t$ in
  $k$ steps, then $f(\vec{s}) \arr{\Ruleslab} f_k(\vec{s})
  \arrr{\Ruleslab} \labl_0(t) = t$ (as there are no defined symbols
  in $t$).
  \qedhere
\end{itemize}
\end{proof}

What is more, as promised, $\arr{\Ruleslab}$ is terminating (even
though $\arr{\Rules}$ might not be).

\begin{lemma}\label{lem:terminating}
There is no infinite $\arrr{\Ruleslab}$ reduction.
\end{lemma}

\begin{proof}
This follows because we can orient all rules by the
\emph{Computability Path Ordering} (CPO)~\cite{bla:jou:rub:08}.
Here, we use only the first definition, without accessibility
(Section 3.3),
%\JGS{Que? You are not referring to the Blanqui et al.\ paper here (as there is no Section 5.1 and ``core'' is not mentioned anywhere in that paper).  You seem to be referring to the LMCS   version of the paper (which just appeared in 2015 where ``core'' CPO is obtained from HORPO by a bit of erasure). Right?}
with the following precedence:
\begin{itemize}
\item for $f,g \in \Defineds,i,j \in \nats$: $f_i \succ_{\Flab}
  g_j$ if $i > j$;
\item for $f \in \Defineds, g \in \Constructors,i \in \nats$:
  $f \succ_{\Flab} f_i \succ_{\Flab} g$.
\end{itemize}
This precedence is obviously well-founded, as there is no infinite
decreasing sequence of numbers in $\nats$.  We employ an order on types
which obeys the requirements and equates all sorts (such an order can
easily be constructed for any given set of sorts).

Observe:
\begin{enumerate}
\item\label{thm:terminating:subterm}
  If $s \in \Terms(\Constructors,\V)$ and $s \suptermeq
  t$, then $s \succeq_\tau t$. \\
  Here, $\succ_\tau$ is the \emph{type-sensitive} part of the ordering,
  so $s : \atype \succ t : \atype$ and $\atype$ is
  greater or equal in the type ordering then $\btype$.  As
  we have assumed that all sorts have a type declaration
  $[\asort_1 \times \dots \times \asort_n] \arrtype \bsort$ with
  $\bsort$ and all $\asort_i$ in $\Sorts$,
  the above observation follows immediately by %($\F_b\rhd$)
case (1e)
  and structural induction on $s$.
\end{enumerate}

Recall that CPO employs a (finite) set of variables $X$ for
bookkeeping related to variables encountered in (above the right-hand
side of) the current constraint to be satisfied. Keeping with standard
notation for CPO \cite{bla:jou:rub:08} we write $s  \succ^X t$ for the
ordering below.
Observe that each rule in $\Ruleslab$ is oriented: the rules with an
unlabeled left-hand side because $f \succ_{\Flab} f_i$ for all $f,i$,
the ``decreasing'' rules $f_{i+1}(\vec{x}) \arrz f_i(\vec{x})$ because
each $f_{i+1} \succ_{\Flab} f_i$, and as for the other rules, we see
by induction that if $r$ is a renaming of a subterm of the right-hand
side of a rule $f(\vec{\ell}) \arrz r \in \Rules$ and $\FV(r)
\setminus \FV(f(\vec{\ell})) \subseteq X$ and only variables not
occurring in $f(\vec{\ell})$ have been renamed, then
$f_{i+1}(\vec{\ell}) \succ^X \labl_i(r)$:
\begin{itemize}
\item if $r$ is a variable in $X$, then $f_{i+1}(\vec{\ell})
  \succ^X r = \labl_i(r)$ by case (1a); %($\F_b\mathcal{X}$);
\item if $r$ is a variable not in $X$, then it occurs in some
  $\ell_j$, so $\ell_j \succeq_\tau r = \labl_i(r)$ by observation
  (\ref{thm:terminating:subterm}), giving $f_{i+1}(\vec{\ell})
  \succ^X \labl_i(r)$ by case (1e). %($F_b\rhd$);
\item if $r = g(r_1,\dots,r_n)$ with $g \in \Defineds$, then
  $\labl_i(r) = g_i(\labl_i(r_1),\dots,\labl_i(r_n))$, and by the
  induction hypothesis $f_{i+1}(\vec{\ell}) \succ^X \labl_i(r_j)$
  for all $j$; we complete by %($\F_b>$)
  case (1c) because $f_{i+1} \succ_{\Flab} g_i$;
\item if $r = g(r_1,\dots,r_n)$ with $g \in \Constructors$,
  then $\labl_i(r) = g(\labl_i(r_1),\dots,\labl_i(r_n))$, and
  by the induction hypothesis $f_{i+1}(\vec{\ell}) \succ^X
  \labl_i(r_j)$ for all $j$; we complete once more by %($\F_b>$)
  case (1c) because $f_{i+1} \succ_{\Flab} g$;
\item if $r = \app{r_1}{r_2}$, then $f_{i+1}(\vec{\ell}) \succ^X
  \labl_i(r_1),\labl_i(r_2)$ by the induction hypothesis, so
  $f_{i+1}(\vec{\ell}) \succ^X \app{\labl_i(r_1)}{\labl_i(r_2)} =
  \labl_i(\app{r_1}{r_2})$ by %($\F_b@$);
  case (1c).
\item if $r = \abs{x}{r'}$, then for a fresh variable $y$,
  $\FV(r'[x:=y]) = \FV(r) \cup \{y\}$; the induction
  hypothesis gives $f_{i+1}(\vec{\ell}) \succ^{X \cup \{y\}}
  \labl_i(r'[x:=y]) = \labl_i(r')[x:=y]$, so we
  obtain $f_{i+1}(\vec{\ell}) \succ^X \labl_i(r)$ by
  %($\F_b\lambda$).
  case (1d).
\end{itemize}
In particular, we thus have $f_{i+1}(\vec{\ell}) \succ_\tau
\labl_i(r)$ for $r$ the right-hand side of the rule.  With all rules
oriented, we obtain well-foundedness of $\arr{\Ruleslab}$
by~\cite[Lemma 6.3 (monotonicity), Lemma 6.6(1) (stability) and
Theorem 6.27 (well-foundedness)]{bla:jou:rub:08}.
\end{proof}

Note that, while we did use Lemma~\ref{lem:niceconstructor} to obtain
that functional variables may only occur as direct arguments of the
root, the proof otherwise does not rely on cons-freeness.

Before turning our attention to Theorem~\ref{thm:algorithmsoundcomplete},
we derive one ancillary lemma:

\begin{lemma}\label{lem:betafirst}
Let $s = \apps{\app{(\abs{x}{u})}{v_0}}{v_1}{v_n}$ with $n \geq 0$ and
$t \in \Data$.  Then $s \arrr{\Ruleslab} t$ if{f}
$\apps{u[x:=v_0]}{v_1}{v_n} \arrr{\Ruleslab} t$.
\end{lemma}

\begin{proof}
For the only if direction, we obtain $s \arr{\Ruleslab}
\apps{u[x:=v_0]}{v_1}{v_n} \arrr{\Ruleslab} t$.  For the if direction,
suppose $s \arrr{\Rules} t$.  As $t \in \Data$ does not contain
applications, the reduction must eventually contract a redex at the root;
we have $s \arrr{\Ruleslab} \apps{\app{(\abs{x}{u'})}{v_0'}}{v_1'}{
v_n'} \arr{\beta} \apps{u'[x:=v_0']}{v_1'}{v_n'} \arrr{\Ruleslab} t$,
with $u \arrr{\Ruleslab} u'$ and each $v_i \arrr{\Ruleslab} v_i'$.
But as $\arr{\Ruleslab}$ is a rewriting relation, and therefore both
monotonic and stable under substitution, also $\apps{u[x:=v_0]}{v_1}{
v_n} \arrr{\Ruleslab} \apps{u'[x:=v_0']}{v_1'}{v_n'} \arrr{\Ruleslab}
t$.
\end{proof}

Note that Lemma~\ref{lem:betafirst} immediately implies that
$t~_\beta\!\leftarrow s \arrr{\Rules} q$ with $q \in \Data$ implies
$t \arrr{\Rules} q$ as well.  Therefore, as announced in the
introduction, we will from now on assume that the rules in $\Rules$
do not contain any $\beta$-redexes, as reducing these immediately
does not change the many-step reduction relation to data which we are
interested in.

As announced in the proof sketch, we will use an auxiliary definition;
the \emph{$\NF$-substitution:}

\begin{definition}
For $V \subseteq \V$, a partial function $\eta$ on domain $V$ is an
$\NF$-substitution of depth $k \geq 0$ if $k$ is the smallest number
such that:
for all $x : \atype \in V$ there exist some $i,s,\zeta$ such that
$\vdash s : \atype$ and $\eta(x) = \NF^i(s,\zeta)$ and $\zeta$ is an
$\NF$-substitution of depth $m < k$.
Note that the empty mapping $[]$ is an $\NF$-substitution of depth
$0$.

For an $\NF$-substitution $\eta$ on domain $V$, let $\overline{\eta}$
be defined by induction on the depth of $\eta$:
\begin{itemize}
\item for $x \notin V$, $\overline{\eta}(x) = x$;
\item for $x \in V$ we can write $\eta(x) = \NF^i(s,\zeta)$ with
  $\mathit{depth}(\zeta) < \mathit{depth}(\eta)$; let
  $\overline{\eta}(x) = \labl_i(s)\overline{\zeta}$.
\end{itemize}
\end{definition}

Now we are ready to prove Theorem~\ref{thm:algorithmsoundcomplete}:

\oldcounter{\soundnessthm}
\begin{theorem}
Let $f : [\asort_1 \times \dots \times \asort_n] \arrtype \bsort \in
\Defineds$ and $s_1 : \asort_1,\dots,s_n : \asort_n, t : \bsort$ be
data terms.  Then $\Conf^I[f(\{s_1\},\dots,\{s_n\}) \approx t] = \top$
if{f} $f(\vec{s}) \arrr{\Rules} t$.
\end{theorem}
\startappendixcounters

\begin{proof}
Extending the definition of $\Conf^i$ and $\NF^i$ also for $i > I$ --
simply by observing that, if the recursive process were continued, we
obtain $\Conf^I = \Conf^{I+1} = \dots$ -- we will derive the following
two statements for all relevant $i,\vec{j} \in \nats,f \in \Defineds,u,
\vec{s} \in \Terms(\F,\V),t \in \B$ and $\NF$-substitutions $\zeta,
\vec{\eta}$:
\begin{description}
\item[(A)]
  $\Conf^i[f(\NF^{j_1}(s_1,\eta_1),\dots,\NF^{j_n}(s_n,\eta_n))
  \approx t] = \top$ \emph{if and only if} \\
  $q := f_i(\labl_{j_1}(s_1)\overline{\eta_1},\dots,\labl_{j_n}(s_n)
  \overline{\eta_n}) \arrr{\Ruleslab} t$;
\item[(B)]
  $t \in \NF^i(u,\zeta)(\NF^{j_1}(s_1,\eta_1),\dots,\NF^{j_n}(s_n,
  \eta_n))$ \emph{if and only if} \\
  $q := \apps{(\labl_i(u)\overline{\zeta})}{\labl_{j_1}(s_1)
  \overline{\eta_1}}{\labl_{j_n}(s_n)\overline{\eta_n}} \arrr{
  \Ruleslab} t$.
\end{description}
If we can prove (A), we obtain the theorem by
Lemma~\ref{lem:labeledequiv}:
\begin{itemize}
\item if $\Conf^I[f(\vec{s}) \approx t] = \top$, we can write this as
  \\
  $\Conf^I[f(\NF^0(s_1,[]),\dots,\NF^0(s_n,[])) \approx t] = \top$,
  which gives \\
  $f(\labl_0(s_1),\dots,\labl_0(s_n)) \arrr{\Ruleslab} t$, so
  $f(s_1,\dots,s_n) \arrr{\Rules} t$ by Lemma~\ref{lem:labeledequiv};
\item if $f(\vec{s}) \arrr{\Rules} t$, then by
  Lemma~\ref{lem:labeledequiv} there is some $i$ with \\
  $f_i(\labl_i(s_1),\dots,\labl_i(s_n)) \arrr{\Ruleslab} t$; then by
  (A) we obtain \\
  $\Conf^i[f(\NF^i(s_1,[]),\dots,\NF^i(s_n,[])) \approx t] = \top$,
  which (because all $s_i \in \Data$) implies \\
  $\Conf^i[f(\{s_1\},\dots,\{s_n\}) \approx t] = \top$.  If $i \leq I$
  then the same holds for $I$ since $\Conf^x[C] = \top$ implies
  $\Conf^{x+1}[C] = \top$, and if $i > I$ this follows because
  $\Conf^I = \Conf^{I+1} = \dots$.
\end{itemize}
We will prove statements (A) and (B) together by a mutual induction on
$q$, oriented with $\arr{\Ruleslab} \cup \rhd$, which is terminating
because $\arr{\Ruleslab}$ is terminating and monotonic.

\smallskip\noindent
\textbf{(A), only if case}.
Suppose $\Conf^i[f(A_1,\dots,A_n) \approx t] = \top$, where $A_k =
\NF^{j_k}(s_k,\eta_k)$ for $1 \leq k \leq n$.  If this holds, then
necessarily $i > 0$; there are two possibilities.
\begin{itemize}
\item $\Conf^{i-1}[f(\vec{A}) \approx t] = \top$.  The induction
  hypothesis immediately yields:
  \[
  \begin{array}{l}
  f_i(\labl_{j_1}(s_1)\overline{\eta_1},\dots, \labl_{j_n}(s_n)
  \overline{\eta_n}) \\ \arr{\Ruleslab}
  f_{i-1}(\labl_{j_1}(s_1)\overline{\eta_1},\dots,
  \labl_{j_n}(s_n)\overline{\eta_n}) \\ \arrr{\Ruleslab} t
  \end{array}
  \]
\item There are a rule $f(\vec{\ell}) \arrz r \in \Rules$
  and substitution $\gamma$ on domain $\FV(f(\vec{\ell}))
  \setminus \{ \vec{\ell} \}$ such that $\ell_k\gamma \in A_k$ for
  all non-variable $\ell_k$ and $t \in \NF^{i-1}(r\gamma,\xi)$,
  where $\xi$ is the function mapping each variable $\ell_k$ to
  $A_k = \NF^{j_k}(s_k,\eta_k)$ -- also an $\NF$-substitution.
  
  Now, for all non-variable $\ell_k$, we use the $\supterm$ part of
  the induction hypothesis (B) to obtain $\labl_{j_k}(s_k)
  \overline{\eta_k} \arrr{\Ruleslab} \ell_k\gamma \in \B$.  Let  
  $\delta := \gamma \cup [\ell_k:=\labl_{j_k}(s_k)\overline{
  \eta_k} \mid \ell_k \in \V]$.  Then we have:
  \[
  \begin{array}{l}
  f_i(\labl_{j_1}(s_1)\overline{\eta_1},\dots,
  \labl_{j_n}(s_n)\overline{\eta_n}) \\ \arrr{\Ruleslab}
  f_i(\ell_1,\dots,\ell_n)\delta \\ \arr{\Ruleslab}
  \labl_{i-1}(r)\delta \\ =
  \labl_{i-1}(r\gamma)[\ell_k:=\labl_{j_k}(s_k)\overline{\eta_k}
  \mid \ell_k \in \V] \\ =
  \labl_{i-1}(r\gamma)\overline{\xi}
  \end{array}
  \]
  Since at least one step is done and $t \in \NF^{i-1}(r\gamma,
  \xi)$, we can use the $\arr{\Ruleslab}$ part of the induction
  hypothesis of (B) to derive that $\labl_{i-1}(r\gamma)\overline{\xi}
  \arrr{\Ruleslab} t$.
\end{itemize}

\smallskip\noindent
\textbf{(A), if case}.  Suppose $q = f_i(\labl_{j_1}(s_1)
\overline{\eta_1},\dots,\labl_{j_n}(s_n)\overline{\eta_n})
\arrr{\Ruleslab} t$.  Since $t$ cannot be rooted by $f_i$, the
reduction must eventually take a root step.  There are two
possibilities.
\begin{itemize}
\item A lowering rule: $q \arrr{\Ruleslab} f_i(x_1,\dots,x_n)\gamma
  \arr{\Ruleslab} f_{i-1}(x_1,\dots,x_n)\gamma \arrr{\Ruleslab} t$.
  Then
  \begin{align*}
  q\ & = f_i(\labl_{j_1}(s_1)\overline{\eta_1},\dots,
  \labl_{j_n}(s_n)\overline{\eta_n}) \\
  & \arr{\Ruleslab} f_{i-1}(\labl_{j_1}(s_1)\overline{\eta_1},\dots,
  \labl_{j_n}(s_n)\overline{\eta_n})\\
  & \arrr{\Ruleslab} f_{i-1}(\vec{x})\gamma \arrr{\Ruleslab} t
  \end{align*}
  By the induction hypothesis, $\Conf^{i-1}[f(\NF^{j_1}(s_1,\eta_1),
  \dots,\NF^{j_n}(s_n,\eta_n)) \approx t] = \top$, so by
  definition the same holds for $\Conf^i[\dots]$.
\item A rule obtained from $\Rules$: $q \arrr{\Ruleslab}
  f_i(\ell_1\gamma,\dots,\ell_n\gamma) \arr{\Ruleslab} \labl_{i-1}(r)
  \gamma \arrr{\Ruleslab} t$ for $f(\vec{\ell}) \arrz r \in \Rules$,
  where each $\labl_{j_k}(s_k)\overline{\eta_k} \arrr{\Rules}
  \ell_k\gamma$.  Now, let $LV := \{ k \mid k \in \{ 1,\dots,n\}
  \wedge \ell_k \in \V \}$.  Let $\delta := [\ell_k:=\labl_{j_k}(s_k)
  \overline{\eta_k} \mid k \in LV]$, and $\gamma' := [x:=\gamma(x)
  \mid x \notin \domain(\delta)]$.  Then:
  \begin{itemize}
  \item $\delta$ and $\gamma'$ have disjoint domains, and
    $\domain(\delta) \cup \domain(\gamma') = \domain(\gamma)$;
  \item $\gamma'$ maps to elements of $\B$, and each $\ell_k\gamma'
    \in \B$ for $k \notin LV$;
  \item each $(\delta \cup \gamma')(x) \arrr{\Ruleslab} \gamma(x)$;
  \item for $k \in \{1,\dots,n\} \setminus LV$: $\labl_{j_k}(s_k)
    \overline{\eta_k} \arrr{\Rules} \ell_k\gamma = \ell_k\gamma' \in
    \B$;
  \item hence, by the induction hypothesis, $\ell_k\gamma' \in
    \NF^{j_k}(s_k,\eta_k)$ for $k \in \{1,\dots,n\}\setminus LV$;
  \item $q \arrr{\Ruleslab} f_i(\vec{\ell})(\delta \cup \gamma')
    \arr{\Ruleslab} \labl_{i-1}(r)(\delta \cup \gamma') \arrr{
    \Ruleslab} \labl_{i-1}(r)\gamma \arrr{\Ruleslab} t$;
  \item $\labl_{i-1}(r)(\delta \cup \gamma') = \labl_{i-1}(r\gamma')
    \delta$ since $\gamma'$ maps to data terms;
  \item $\delta = \overline{\chi}$, where $\chi = [\ell_k := \NF^{j_k
    }(s_k,\eta_k) \mid k \in LV]$;
  \item thus, by the induction hypothesis, $\labl_{i-1}(r\gamma')
    \overline{\chi} \arrr{\Rules} t$ implies $t \in \NF^{i-1}(r\gamma
    ',\chi)$;
  \item this gives $\Conf^i[f(\NF^{j_1}(s_1,\eta_1),\dots,\NF^{j_n}(
    s_n,\eta_n)) \approx t] = \top$.
  \end{itemize}
\end{itemize}
\smallskip\noindent
\textbf{(B), both cases}.  We prove (B) by two additional induction
hypotheses; the second on the depth of $\xi$, the third on the size
of $u$.  Consider the form of $u$.
\begin{itemize}
\item If $u = f(u_1,\dots,u_m)$, then $u$ has base type, so $n =
  0$ and $t \in \NF^i(u,\xi)$ if and only if $\Conf^i[f(\NF^i(u_1,
  \xi),\dots,\NF^i(u_m,\xi)) \approx t] = \top$.  As we have just
  seen, this is the case if{f} $q =
  \labl_i(u)\overline{\xi} =
  f_i(\labl_i(u_1)\overline{\xi},\dots,\labl_i(u_m)
  \overline{\xi}) \arrr{\Ruleslab} t$.
\item If $u \in \V$, then since $\domain(\xi) \supseteq
  \FV(u)$ we can write $\xi(u) = \NF^{i'}(u',\xi')$,
  and have
  \[
  \begin{array}{l}
  \NF^i(u,\xi)(\NF^{j_1}(s_1,\eta_1),\dots,\NF^{j_n}(
  s_n,\eta_n)) \\
  = \NF^i(u',\xi')(\NF^{j_1}(s_1,\eta_1),\dots,\NF^{j_n}(
  s_n,\eta_n))
  \end{array}
  \]
  Also,
  \[
  \begin{array}{l}
  \apps{(\labl_i(u)\overline{\xi})}{(\labl_{j_1}(s_1)
  \overline{\eta_1})}{(\labl_{j_n}(s_n)\overline{\eta_n})} \\
  = \apps{\overline{\xi}(u)}{(\labl_{j_1}(s_1)
  \overline{\eta_1})}{(\labl_{j_n}(s_n)\overline{\eta_n})} \\
  = \apps{(\labl_{i'}(u')\overline{\xi'})}{(\labl_{j_1}(s_1)
  \overline{\eta_1})}{(\labl_{j_n}(s_n)\overline{\eta_n})} \\
  \end{array}
  \]
  Noting that $\zeta'$ has a smaller depth than $\zeta$, we complete
  by the second induction hypothesis.
\item If $u = \app{v}{w}$, then
  \[
  \begin{array}{l}
  \NF^i(u,\xi)(\NF^{j_1}(s_1,\eta_1),\dots,\NF^{j_n}(s_n,\eta_n)) \\
  = \NF^i(v,\xi)(\NF^i(w,\xi),
  \NF^{j_1}(s_1,\eta_1),\dots,\NF^{j_n}(s_n,\eta_n))
  \end{array}
  \]
  Additionally,
  \[
  \begin{array}{l}
  \apps{(\labl_i(u)\overline{\xi})}{(\labl_{j_1}(s_1)
  \overline{\eta_1})}{(\labl_{j_n}(s_n)\overline{\eta_n})} \\
  = \apps{\app{(\labl_i(v)\overline{\xi})}{(\labl_i(w)
  \overline{\xi})}}{(\labl_{j_1}(s_1)
  \overline{\eta_1})}{(\labl_{j_n}(s_n)\overline{\eta_n})}
  \end{array}
  \]
  We complete by the third induction hypothesis.
\item Finally, if $u = \abs{x}{u'}$, then $n > 0$ by type
  restrictions.  Then
  \[
  \begin{array}{l}
  \NF^i(u,\xi)(\NF^{j_1}(s_1,\eta_1),\dots,\NF^{j_n}(s_n,\eta_n)) \\
  = \NF^i(u',\xi \cup [x:=\NF^{j_1}(s_1,\eta_1)])(\NF^{j_2}(s_2,
  \eta_2),\dots,\NF^{j_n}(s_n,\eta_n))
  \end{array}
  \]
  Now, assuming $x$ to be fresh (which we can safely do by
  $\alpha$-conversion), $\delta := \xi \cup [x :=
  \NF^{j_1}(s_1,\eta_1)]$ is an $\NF$-substitution.  We note that:
  \begin{align*}
  q\ 
  & = \apps{(\labl_i(\abs{x}{u'})\overline{\xi})}{(\labl_{j_1}(s_1)
  \overline{\eta_1})}{(\labl_{j_n}(s_n)\overline{\eta_n})} \\
  & = \apps{(\abs{x}{(\labl_i(u')\overline{\xi})})}{(\labl_{j_1}(s_1)
  \overline{\eta_1})}{(\labl_{j_n}(s_n)\overline{\eta_n})} \\
  & \arr{\beta} \apps{(\labl_i(u')\overline{\xi}[x:=\labl_{j_1}(s_1,
  \overline{\eta_1})])}{(\labl_{j_2}(s_2)
  \overline{\eta_2})}{(\labl_{j_n}(s_n)\overline{\eta_n})} \\
  & = \apps{(\labl_i(u')\overline{\delta})}{(\labl_{j_2}(s_2)
  \overline{\eta_2})}{(\labl_{j_n}(s_n)\overline{\eta_n})} \\
  & =: q'
  \end{align*}
  As $q$ reduces to $q'$, we use the first induction hypothesis to
  obtain $q' \arrr{\Ruleslab} t$ if{f} $t \in \NF^i(u',\xi
  \cup [x:=\NF^{j_1}(s_1,\eta_1)])(\NF^{j_2}(s_2,\eta_2),\dots,
  \NF^{j_n}(s_n,\eta_n))$.  This proves the theorem since
  Lemma~\ref{lem:betafirst} gives us that $q \arrr{\Ruleslab} t$ if
  and only if $q' \arrr{\Ruleslab} t$.
\qedhere
\end{itemize}
\end{proof}

Finally, we consider the proof of complexity.  Again, we split this
into several parts.  To start, we define a counter notion to
\emph{order}:

\begin{definition}
The \emph{length bound} of a type $\atype$ is the length $n+1$ of the
longest sequence $\atype_1 \arrtype \dots \arrtype \atype_n \arrtype
\asort$ occurring in it.  Formally, $\mathit{lengthbound}(\atype_1
\arrtype \dots \arrtype \atype_n \arrtype \asort) = \max(n+1,
\mathit{lengthbound}(\atype_1),\dots,\mathit{lengthbound}(\atype_n))$.
\end{definition}

\begin{lemma}\label{lem:interpretsize}
If a type $\atype$ has order $k$ and length bound at most $i$, then
$\card(\interpret{\atype}) \leq \mathrm{exp}_2^{k+1}(i^{k+1} \cdot
N)$, where $N$ is the number of elements in $\B$.
\end{lemma}

\begin{proof}
By induction on the form of $\atype$; write $\atype = \atype_1
\arrtype \dots \arrtype \atype_n \arrtype \asort$ with $0 \leq n < i$
and each $\atype_i$ having order at most $k-1$ and length bound at
most $i$.  Then:
\[
\begin{array}{rcl}
\card(\interpret{\atype_1 \arrtype \dots \arrtype \atype_n \arrtype
\asort}) & = &
\card((\cdots (\interpret{\asort}^{\interpret{\atype_n}})^{
\interpret{\atype_{n-1}}}\cdots)^{\interpret{\atype_1}}) \\
& = &
(\cdots (\card(\interpret{\asort})^{\card(\interpret{\atype_n})})^{
\card(\interpret{\atype_{n-1}})}\cdots)^{\card(\interpret{\atype_1})}
\\
& = & \card(\interpret{\asort})^{\card(\interpret{\atype_n})\cdots
  \card(\interpret{\atype_1})} \\
& \leq & 2^{N \cdot \card(\interpret{\atype_n}) \cdots
  \card(\interpret{\atype_1})}\hfill
  (\text{since}\ \interpret{\asort} \subseteq \P(\B)) \\
& \leq & 2^{N \cdot \mathrm{exp}_2^k(i^k \cdot N) \cdots
  \mathrm{exp}_2^k(i^k \cdot N)}\hfill
  (\text{by induction hypothesis}) \\
& = & 2^{N \cdot \mathrm{exp}_2^k(i^k \cdot N)^n} \\
& \leq & 2^{\mathrm{exp}_2^k(i^k \cdot N \cdot n + N)}\hfill
  (\text{by (**)}) \\
& = & \mathrm{exp}_2^{k+1}(n \cdot i^k \cdot N + N) \\
& \leq & \mathrm{exp}_2^{k+1}(i \cdot i^k \cdot N)\ \ \ 
  (\text{as}\ n \cdot i^k + 1 \leq (n+1) \cdot i^k \leq i \cdot i^k)
\\ & = & \exp_2^{k+1}(i^{k+1} \cdot N) \\
\end{array}
\]
(**) Here, we make an additional claim: $N \cdot \mathrm{exp}_2^k(m
\cdot N)^n \leq \mathrm{exp}_2^k(m \cdot N \cdot n + N)$ for $m,k
\geq 1$ and all $X$. This claim obviously holds if $N = 0$; for $N >
0$ we prove it by induction on $k$:
\begin{itemize}
\item if $k = 1$ then $N \cdot (2^{m \cdot N})^n = N \cdot 2^{m \cdot
  N \cdot n} \leq 2^N \cdot 2^{m \cdot N \cdot n} = 2^{N + m \cdot N
  \cdot n}$;
\item if the claim is known for $k$, then $N \cdot \mathrm{exp}_2^{
  k+1}(m \cdot N)^n = N \cdot \mathrm{exp}_2^k(2^{m \cdot N})^n \leq
  \mathrm{exp}_2^k(2^{m \cdot N} \cdot n + N)$; we are done if we can
  prove that $2^{m \cdot N} \cdot n + N \leq 2^{m \cdot N \cdot n +
  N}$, which holds because always:
  \begin{itemize}
  \item $2^X \cdot n \leq 2^{X \cdot n}$ when $X \geq 1$:
    if $n = 0$ both sides are $0$, if $n = 1$ both sides are $2^X$,
    if $n \geq 2$ and $X = 1$ the statement becomes $2n \leq 2^n$
    which indeed holds for $n \geq 2$, and if $n \geq 2$ and $X \geq
    2$ we obtain $2^X \cdot n \leq 2^X \cdot 2^n = 2^{X + n} \leq
    2^{X \cdot n}$;
  \item $2^X + N \leq 2^X + 2^X \cdot N = 2^X \cdot (1 + N) \leq
    2^X \cdot 2^N = 2^{X + N}$.
    \qedhere
  \end{itemize}
\end{itemize}
\end{proof}

Lemma~\ref{lem:interpretsize} bounds the sizes of the sets
iterated over in the algorithm. Preparations done, consider the theorem:

\oldcounter{\complexitythm}
\begin{theorem}
If $(\F,\Rules)$ has order $k$, then Algorithm~\ref{alg:main}
runs in time $O(\mathrm{exp}_2^k(m \cdot n))$ for some $m$.
\end{theorem}
\startappendixcounters

\begin{proof}
In the following, denote by ``the set of types occurring in an AFS'' $(\F,
\Rules)$ the set $\Sigma$ of all $\atype_1,\dots,\atype_n,\asort$
such that some $f : [\atype_1 \times \dots \times \atype_n] \arrtype
\asort \in \F$, and all their subtypes.  We let:
\begin{itemize}
\item $a \in \nats$ denote the maximal arity of symbols in $\F$;
\item $k \in \nats$ denote the order of the AFS $(\F,\Rules)$, so $k-1$ the
  maximal type order in $\Sigma$;
\item $i \in \nats$ denote a length bound which bounds all $\atype$ in
  $\Sigma$;
\item $d \in \nats$ denote the maximal size (counting symbols,
  variables, applications and abstractions) of right-hand sides in
  $\Rules$.
\end{itemize}
All numbers above are fixed by the given AFS and should thus be
considered constant (the only input to the algorithm is $s$).  We
also define:
\begin{itemize}
\item $N :=$ the number of elements in $\B$ (note that this number is
  linear in $|s|$);
\item $X := \mathrm{exp}_2^k(i^k \cdot N)$, which bounds
  $\card(\atype)$ for all $\atype \in \Sigma$ by
  Lemma~\ref{lem:interpretsize};
\item $Y := |\Defineds| \cdot X^a \cdot N \leq
  |\Defineds| \cdot \mathrm{exp}_2^k((i^k + a + 1) \cdot N)$, which
  therefore bounds the number of different statements $f(\vec{A})
  \approx t$ considered in the algorithm;
\end{itemize}
Now, for every right-hand side $r$, we first make the following
observation: every subterm of $r$ has a type which is a sort or in
$\Sigma$.  This follows because we have assumed right-hand sides to
be $\beta$-normalised, so all strict subterms are either the direct
argument of some $f \in \F$ or of an application $\apps{F}{r_1}{r_n}$
with $F$ a variable which occurs as a direct argument in the
left-hand side.  Thus, in particular, the binders of abstractions
have a type of at most order $k-2$.

Consider the cost of calculating some $\NF^i(r\gamma,\eta)$ if
all $\gamma(x)$ are data terms (or variables) and $\Conf^i$ is already
known; the exact cost depends on implementation details, so for
simplicity let $Z$ denote a bound to the cost of:
\begin{itemize}
\item performing a substitution $r'\gamma$;
\item looking up a truth value in $\Conf^i$ if $A_1,\dots,A_n,t$ are
  already calculated;
\item looking up an element in $\eta$;
\item calculating a function $A(B)$, with $A \in \interpret{\atype
  \arrtype \btype},\ B \in \interpret{\atype}$ for some $\atype,
  \btype \in \Sigma$.
\end{itemize}
Induction on the size $|r|$ of $r$ shows that the cost of calculating
$\NF^i(r\gamma,\eta)$ is bounded by $Z \cdot X^{|r|} \cdot |r|$:
\begin{itemize}
\item if $r = c(\dots)$ with $c \in \Constructors$, or $r$ is a
  variable in $\domain(\gamma)$, this cost is at most $Z$;
\item if $r$ is a variable in $\domain(\eta)$, this cost is at most
  $Z$;
\item if $r = \app{u}{v}$, we must calculate $\NF^i(u\gamma,\eta)$
  and $\NF^i(v\gamma,\eta)$, followed by a function calculation, so
  this cost is at most
  \[
  \begin{array}{rcl}
  (Z \cdot X^{|r_1|} \cdot |r_1|) + (Z \cdot X^{|r_2|} \cdot |r_2|) +
    Z & \leq & (Z \cdot X^{|r|} \cdot |r_1|) + (Z \cdot X^{|r|} \cdot
    |r_2|) + (Z \cdot X^{|r|} \cdot 1) \\
  & = & (Z \cdot X^{|r|}) \cdot (|r_1| + |r_2| + 1) \\
  & = & Z \cdot X^{|r|} \cdot |r|
  \end{array}
  \]
\item if $r = f(r_1,\dots,r_n)$, then we must calculate $\NF^i(r_i
  \gamma,\eta)$ for each subterm, so we obtain a cost bounded by
  \[
  \begin{array}{rcl}
  (\sum_{i=1}^n Z \cdot X^{|r_i|} \cdot |r_i|) + Z \cdot N &
    \leq &
  (\sum_{i=1}^n Z \cdot X^{|r|} \cdot |r_i|) + (Z \cdot X^{|r|}
    \cdot 1) \\
  & = & Z \cdot X^{|r|} \cdot (|r_1| + \dots |r_n| + 1) \\
  & = & Z \cdot X^{|r|} \cdot |r| \\
  \end{array}
  \]
\item if $r = \abs{x}{r'}$, then there are fewer than $X$ different
  $\NF^i(r',\zeta)$ to calculate, so the cost is bounded by
  $X \cdot (Z \cdot X^{|r'|} \cdot |r'|) \leq
  Z \cdot X^{|r'|+1} \cdot |r'| \leq Z \cdot X^{|r|} \cdot |r|$.
\end{itemize}
Even in a non-optimal
implementation, we will have $Z \leq c \cdot Y^b$ for some $b,c$, which suffices
for our purposes. This bounds the cost of determining a query $t \in
\NF^i(r\gamma,\eta)$ by $c \cdot d \cdot Y^b \cdot X^d$.

Observing that $I \leq Y+2$, as the number of $\top$-statements
increases by at least $1$ in every step before $I$, we thus obtain:
\begin{itemize}
\item there are at most $Y+2$ steps;
\item in each step, we investigate at most $Y$ claims;
\item for each claim, we consider $|\Rules|$ possible rules;
\item for each rule, we investigate at most $(2^N)^a = 2^{a \cdot
  N}$ substitutions $\gamma$;
\item for each investigation, we must test membership in some
  $\NF^i(r\gamma,\eta)$.
\end{itemize}
The cost of the lookup to $\Conf^{i-1}$ is negligible compared to the cost of investigating all
substitutions.  Combining these costs and assuming $N \geq 1$, we
obtain a bound of
\[
\begin{array}{rl}
& (Y+2) \cdot Y \cdot |\Rules| \cdot (2^{a \cdot N} + 1) \cdot c
  \cdot d \cdot Y^b \cdot X^d \\
= & (|\Defineds| \cdot X^a \cdot N+2) \cdot |\Defineds| \cdot X^a
  \cdot N \cdot |\Rules| \cdot (2^{a \cdot N} + 1) \cdot c \cdot d
  \cdot (|\Defineds| \cdot X^a \cdot N)^b \cdot X^d \\
\leq & (2 \cdot |\Defineds| \cdot X^a \cdot N) \cdot |\Defineds|
  \cdot X^a \cdot N \cdot |\Rules| \cdot (2^{a \cdot N + 1}) \cdot c
  \cdot d \cdot (|\Defineds|^b \cdot X^{a \cdot b} \cdot N^b) \cdot
  X^d \\
= & 2 \cdot c \cdot d \cdot |\Defineds|^{2+b} \cdot |\Rules| \cdot
  X^{2a+ab + d} \cdot 2^{a \cdot N + 1} \cdot N^{2+b} \\
= & O(\mathrm{exp}_2^k(i^k \cdot N)^{2a + ab + d} \cdot 2^{a \cdot N
  + 1} \cdot 2^{N \cdot (2 + b)}) \\
\leq & O(\mathrm{exp}_2^k(i^k \cdot (3a+ab+d+2) \cdot N + 2b+1)) \\
= & O(\mathrm{exp}_2^k(x \cdot N + y))\ \text{for fixed numbers}\ x\ 
  \text{and}\ y \\
\end{array}
\]
As $N$ is linear in the size of the input, the result follows.
\end{proof}

\newpage

\section{An extended example of SAT-solving using cons-free rewriting}\label{app:sat}

To see how the algorithm from Figure~\ref{fig:sat} works in
practice, consider the formula
$(x_1 \vee x_2) \wedge (\neg x_1 \vee \neg x_2 \vee \neg x_3) \wedge
(x_2 \vee x_3)$.
%$(x_1 \vee x_2) \wedge \neg x_1$.
This corresponds to the following string and term:
\[
L = 11?\#000\#?11\#\quad\quad
\encode{L} =
\one(\one(\symb{?}(\symb{\#}(\nul(\nul(\nul(\symb{\#}(\symb{?}(
\one(\one(\symb{\#}(\nil))))))))))))
\]
We consider a successful reduction from $\symb{decide}(\encode{L})$ to
$\symb{true}$.  For readability, we will omit the brackets and $\nil$,
and simply denote $\encode{L}$ as $\symb{11?\#000\#?11}$ (and similar
for its subterms).
\[
\begin{array}{ll}
& \symb{decide}(\encode{L}) \\
\arr{\Rules} & \symb{assign}(\symb{11?\#000\#?11},~\nil,~\nil,~
  \encode{L}) \\
\arr{\Rules} & \symb{assign}(\symb{1?\#000\#?11},~\nil,~
  \symb{either}(\symb{11?\#000\#?11},\nil),~\encode{L}) \\
\arr{\Rules} & \symb{assign}(\symb{?\#000\#?11},~
  \symb{either}(\symb{1?\#000\#?11},\nil),~
  \symb{either}(\symb{11?\#000\#?11},\nil),~\encode{L}) \\
\arr{\Rules} & \symb{assign}(\symb{\#000\#?11},~
  \symb{either}(\symb{1?\#000\#?11},\nil),\\
  & \phantom{\symb{assign}(\symb{\#000\#?11},}~~
  \symb{either}(\symb{?\#000\#?11},\symb{either}(
    \symb{11?\#000\#?11},\nil)),~\encode{L}) \\
\arr{\Rules} & \symb{main}(\symb{either}(\symb{1?\#000\#?11},\nil),~
  \symb{either}(\symb{?\#000\#?11},\symb{either}(\symb{11?\#000\#?11},
  \nil)),~\encode{L})
\end{array}
\]
This derivation corresponds to choosing the assignment $[x_1:=\bot,
x_2:=\top,x_3:=\bot]$.
For brevity, let us write $X_2$ for
$\symb{either}(\symb{1?\#000\#?11},\nil)$ and $X_{3,1}$ for
$\symb{either}(\symb{?\#000\#?11},\linebreak
\symb{either}(\symb{11?\#000\#?11},\nil))$.  Then
$X_{1} \arrr{\Rules} \symb{1?\#000\#?11}$ and both
$X_{3,1} \arrr{\Rules} \symb{?\#000\#?11}$ and
$X_{3,1} \arrr{\Rules} \symb{11?\#000\#?11}$ using the
$\symb{either}$ rules.  Technically, both terms also reduce to
$\nil$, but we will not use this.

We continue:
\[
\begin{array}{ll}
& \symb{main}(X_2,~X_{3,1},~\symb{11?\#000\#?11\#}) \\
\arr{\Rules} & \symb{test}(X_2,~X_{3,1},~\symb{1?\#000\#?11\#}, \\
  & \phantom{test(}~
  \symb{eq}(X_2,\symb{11?\#000\#?11\#}),~
  \symb{eq}(X_{3,1},\symb{11?\#000\#?11\#})) \\
\arrr{\Rules} & \symb{test}(X_2,~X_{3,1},~\symb{1?\#000\#?11\#},~
  \symb{eq}(\dots),~\symb{eq}(\symb{11?\#000\#?11\#},
  \symb{11?\#000\#?11\#})) \\
\arr{\Rules} & \symb{test}(X_2,~X_{3,1},~\symb{1?\#000\#?11\#},~
  \symb{eq}(\dots),~\symb{eq}(\symb{1?\#000\#?11\#},
  \symb{1?\#000\#?11\#})) \\
\arr{\Rules} & \symb{test}(X_2,~X_{3,1},~\symb{1?\#000\#?11\#},~
  \symb{eq}(\dots),~\symb{eq}(\symb{?\#000\#?11\#},
  \symb{?\#000\#?11\#})) \\
\arr{\Rules} & \symb{test}(X_2,~X_{3,1},~\symb{1?\#000\#?11\#},~
  \symb{eq}(\dots),~\symb{eq}(\symb{\#000\#?11\#},
  \symb{\#000\#?11\#})) \\
\arr{\Rules} & \symb{test}(X_2,~X_{3,1},~\symb{1?\#000\#?11\#},~
  \symb{eq}(\dots),~\strue) \\
\arr{\Rules} & \symb{main}(X_2,~X_{3,1},~\symb{1?\#000\#?11\#}) \\
\end{array}
\]
That is, we tested the first variable of the first clause $x_1 \vee
x_2$ against our non-deterministically chosen assignment, and concluded that it
does not suffice (since $x_1$ is mapped to $\bot$, as evidenced by
$X_{3,1} \arrr{\Rules} \symb{11?\#000\#?11\#}$).  We continue with
the next variable:
\[
\begin{array}{ll}
& \symb{main}(X_2,~X_{3,1},~\symb{1?\#000\#?11\#}) \\
\arr{\Rules} & \symb{test}(X_2,~X_{3,1},~\symb{?\#000\#?11\#},\\
& \phantom{\symb{test}(}~
  \symb{eq}(X_2,\symb{1?\#000\#?11\#}),~
  \symb{eq}(X_{3,1},\symb{1?\#000\#?11\#})) \\
\arrr{\Rules} & \symb{test}(X_2,~X_{3,1},~\symb{?\#000\#?11\#},~
  \symb{eq}(\symb{1?\#000\#?11\#},\symb{1?\#000\#?11\#}),~
  \symb{eq}(\dots)) \\
\arr{\Rules} & \symb{test}(X_2,~X_{3,1},~\symb{?\#000\#?11\#},~
  \symb{eq}(\symb{?\#000\#?11\#},\symb{?\#000\#?11\#}),~
  \symb{eq}(\dots)) \\
\arr{\Rules} & \symb{test}(X_2,~X_{3,1},~\symb{?\#000\#?11\#},~
  \symb{eq}(\symb{\#000\#?11\#},\symb{\#000\#?11\#}),~
  \symb{eq}(\dots)) \\
\arr{\Rules} & \symb{test}(X_2,~X_{3,1},~\symb{?\#000\#?11\#},~
  \strue,~\symb{eq}(\dots)) \\
\arr{\Rules} & \symb{main}(X_2,~X_{3,1},~\symb{skip}(
  \symb{?\#000\#?11\#})) \\
\arr{\Rules} & \symb{main}(X_2,~X_{3,1},~\symb{skip}(
  \symb{\#000\#?11\#})) \\
\arr{\Rules} & \symb{main}(X_2,~X_{3,1},~\symb{000\#?11\#}) \\
\end{array}
\]
Thus, testing the second variable ($x_2$) against our assignment
succeeded, so the $\symb{main}$ function moves on to the next clause.
\[
\begin{array}{ll}
& \symb{main}(X_2,~X_{3,1},~\symb{000\#?11\#}) \\
\arr{\Rules} & \symb{test}(X_2,~X_{3,1},~\symb{00\#?11\#},~
  \symb{eq}(X_{3,1},\symb{000\#?11\#}),~
  \symb{eq}(X_2,\symb{000\#?11\#})) \\
\arrr{\Rules} & \symb{test}(X_2,~X_{3,1},~\symb{00\#?11\#},~
  \symb{eq}(\symb{11?\#000\#?11\#},\symb{000\#?11\#}),~
  \symb{eq}(\dots)) \\
\arr{\Rules} & \symb{test}(X_2,~X_{3,1},~\symb{00\#?11\#},~
  \symb{eq}(\symb{1?\#000\#?11\#},\symb{00\#?11\#}),~
  \symb{eq}(\dots)) \\
\arr{\Rules} & \symb{test}(X_2,~X_{3,1},~\symb{00\#?11\#},~
  \symb{eq}(\symb{?\#000\#?11\#},\symb{0\#?11\#}),~
  \symb{eq}(\dots)) \\
\arr{\Rules} & \symb{test}(X_2,~X_{3,1},~\symb{00\#?11\#},~
  \symb{eq}(\symb{\#000\#?11\#},\symb{\#?11\#}),~
  \symb{eq}(\dots)) \\
\arr{\Rules} & \symb{test}(X_2,~X_{3,1},~\symb{00\#?11\#},~
  \strue,~\symb{eq}(\dots)) \\
\arr{\Rules} & \symb{main}(X_2,~X_{3,1},~\symb{skip}(
  \symb{00\#?11\#})) \\
\arr{\Rules} & \symb{main}(X_2,~X_{3,1},~\symb{skip}(
  \symb{0\#?11\#})) \\
\arr{\Rules} & \symb{main}(X_2,~X_{3,1},~\symb{skip}(
  \symb{\#?11\#})) \\
\arr{\Rules} & \symb{main}(X_2,~X_{3,1},~\symb{?11\#}) \\
\end{array}
\]
The second clause was satisfied already by the valuation for $x_1$,
so the reduction has moved towards the last clause.  Note that here
$\symb{eq}(\symb{11?\#000\#?11\#},\symb{000\#?11\#})$ reduces to
$\strue$, because what is compared is not the exact string, but
rather the number of symbols before the first $\symb{\#}$.
From the current state, we quickly complete the derivation:
\[
\begin{array}{ll}
& \symb{main}(X_2,~X_{3,1},~\symb{?11\#}) \\
\arr{\Rules} & \symb{main}(X_2,~X_{3,1},~\symb{11\#}) \\
\arr{\Rules} & \symb{test}(X_2,~X_{3,1},~\symb{1\#},~
  \symb{eq}(X_2,\symb{11\#}),~\symb{eq}(X_{3,1},\symb{11\#})) \\
\arr{\Rules} & \symb{test}(X_2,~X_{3,1},~\symb{1\#},~
  \symb{eq}(\symb{1?\#000\#?11\#},\symb{11\#}),~\symb{eq}(\dots)) \\
\arr{\Rules} & \symb{test}(X_2,~X_{3,1},~\symb{1\#},~
  \symb{eq}(\symb{?\#000\#?11\#},\symb{1\#}),~\symb{eq}(\dots)) \\
\arr{\Rules} & \symb{test}(X_2,~X_{3,1},~\symb{1\#},~
  \symb{eq}(\symb{\#000\#?11\#},\symb{\#}),~\symb{eq}(\dots)) \\
\arr{\Rules} & \symb{test}(X_2,~X_{3,1},~\symb{1\#},~\strue,~
  \symb{eq}(\dots)) \\
\arr{\Rules} & \symb{main}(X_2,~X_{3,1},~\symb{skip}(\symb{1\#})) \\
\arr{\Rules} & \symb{main}(X_2,~X_{3,1},~\symb{skip}(\symb{\#})) \\
\arr{\Rules} & \symb{main}(X_2,~X_{3,1},~\nil) \\
\arr{\Rules} & \strue \\
\end{array}
\]
\end{document}